\documentclass[journal,onecolumn]{IEEEtran}
\ifCLASSINFOpdf
\else
\fi
\usepackage{amsmath, amsthm, amssymb, amsfonts, graphicx, subfigure, cases, braket, enumerate, booktabs, threeparttable, tabularx,makecell,mathtools,cite}
\usepackage[colorlinks,citecolor=blue,linkcolor=blue, urlcolor=blue, anchorcolor=blue]{hyperref}
\usepackage{braket}
\usepackage{multirow} 
\usepackage[ruled,linesnumbered]{algorithm2e}
\usepackage{enumerate}
\usepackage{tikz}
\usetikzlibrary{quantikz}  
\usepackage{adjustbox}
\usepackage{authblk}
\allowdisplaybreaks[4] 
\usepackage[T1]{fontenc}

\newtheorem{lemma}{Lemma}
\newtheorem{theorem}{Theorem}
\newtheorem{corollary}{Corollary}
\newtheorem{definition}{Definition}
\newtheorem{fact}{Fact}

\newtheorem{proposition}{Proposition}
\newtheorem{property}{Property}

\theoremstyle{definition}

\begin{document}
%
\title{Agnostic learning of qudit stabilizer states}
%
%
%

\author{Wentao Qi,        
Boyan Xu,        
Shiguang Feng,        
and Lvzhou Li
\thanks{This work was supported by the Quantum Science and Technology-National Science and Technology Major Project(Grant No. 2024ZD0300500), the National Natural Science Foundation of China (Grant No. 92465202, 62272492, 12447107),  the Guangdong Provincial Quantum Science Strategic Initiative (Grant No. GDZX2303007, GDZX2503001), the Guangzhou Science and Technology Program (Grant No. 2024A04J4892).}
\thanks{Wentao Qi and Boyan Xu are co-first authors of the article. Lvzhou Li is the corresponding author. Wentao Qi, Boyan Xu, Shiguang Feng and Lvzhou Li are with the Institute of Quantum Computing and Software, School of Computer Science and Engineering, Sun Yat-sen University, Guangzhou 510006, China (email: qiwt5@mail2.sysu.edu.cn; xuby26@mail2.sysu.edu.cn; fengshg3@mail.sysu.edu.cn,lilvzh@mail.sysu.edu.cn).}
}

%
%

\markboth{Journal of \LaTeX\ Class Files,~Vol.~14, No.~8, August~2015}%
{Shell \MakeLowercase{\textit{et al.}}: Bare Demo of IEEEtran.cls for IEEE Journals}
%



\maketitle

\begin{abstract}
Learning a classical description of a quantum state is a fundamental task in quantum computation. Among the most important classes of quantum states are stabilizer states, which play a central role in quantum error correction and fault-tolerant computation. To mitigate the effects of realistic noise, agnostic learning of stabilizer states has emerged as a natural and well-motivated problem. Recently, Chen \textit{et~al.} [STOC ’25, p. 429–438] resolved this problem for qubit systems by using a stabilizer bootstrapping framework. However, the agnostic learning of qudit stabilizer states remains largely unexplored, since the qudit setting introduces fundamental structural differences that prevent a direct generalization of existing qubit techniques. In this paper, we successfully generalize the stabilizer bootstrapping framework to qudit systems and present the first efficient quantum algorithm for agnostic learning of qudit stabilizer states. Specifically, given copies of an unknown $n$-qudit pure state $\ket{\psi}$ that has fidelity $\tau$ with some stabilizer state, our algorithm outputs a stabilizer state $\ket{\phi}$ such that $\left| \braket{\phi|\psi} \right|^2 \geq \tau - \varepsilon$ with high probability. The algorithm uses only single-copy and four-copy measurements, and its sample and time complexity scale as $(d/\tau)^{O(d^2 \log(1/\tau))} \cdot \mathrm{poly}(n, 1/\varepsilon)$,  where the dimension  $d$ is an odd prime. As a direct corollary, our algorithm enables efficient estimation of the magic of a quantum state, as quantified by its stabilizer fidelity. Completing the picture, we also present a streamlined algorithm for the high-fidelity regime $\tau > \cos^2(\pi/8)$, establishing a qudit analogue of the threshold-based approach in prior qubit work.
\end{abstract}


%
\IEEEpeerreviewmaketitle

\section{Introduction}
%
%
%
%

\IEEEPARstart{A}{} fundamental task in quantum computation is efficiently learning classical descriptions of quantum states. For arbitrary, unstructured states, this task suffers from the curse of dimensionality, as the required resources scale exponentially with the system size~\cite{HHJWY16,OW16}. To circumvent this barrier, a growing body of work has demonstrated that when the unknown state exhibits a specific structure, such as phase states~\cite{ABDY23}, matrix product states~\cite{LCLP10,MPF10}, and states prepared by Clifford circuits doped with a few non-Clifford gates~\cite{GIKL25,CLL24,LOH24}, learning its description requires only polynomial sample and time complexity.

Unfortunately, the algorithms developed in these works only work in the realizable setting. Specifically, they  come with a stringent requirement that the unknown quantum state must exactly match the hypothetical structure. In practice, however, quantum systems are inevitably affected by noise, which drives the target state away from the ideal structure and thus renders existing algorithms inapplicable. This challenge has motivated the study of the agnostic learning model, where we no longer require the unknown state to belong exactly to a specific class $\mathcal{C}$. Rather, we aim to find the state within $\mathcal{C}$ that best approximates the target. To formalize this notion, we present the following definition.
\begin{definition}[Agnostic learning of quantum states]
Let $\varepsilon,\delta>0$, and $\mathcal{C}$ a class of quantum states. Given access to copies of an unknown state $\ket{\psi}$, the goal is to learn a classical description of a state $\ket{\phi}$ such that, with probability at least $1-\delta$, the fidelity $F(\ket{\psi},\ket{\phi})$ satisfies
\begin{equation}
F(\ket{\psi},\ket{\phi}) \geq \max_{\ket{\phi'}\in \mathcal{C}} F(\ket{\psi},\ket{\phi'})-\varepsilon.
\end{equation}
If the algorithm outputs a state $\ket{\phi} \in \mathcal{C}$, then it is called proper; otherwise, it is called improper.
\end{definition}

A canonical example of such a class $\mathcal{C}$ is the set of stabilizer states. An $n$-qubit state is said to be a stabilizer state if it is stabilized by a group of $2^n$ commuting Pauli operators. Stabilizer states are among the most important structured classes of quantum states, finding widespread applications in quantum error correction~\cite{Sho95,CS96,Got96}, efficient classical simulation of quantum circuits~\cite{AG04,BSS16,BBC19}, randomized benchmarking~\cite{KLR08,MGE11,HWFW19}, and quantum learning theory~\cite{HKP20}. Agnostic learning of stabilizer states has therefore attracted considerable attention, leading to several efficient algorithms~\cite{GIKL24b, CGYZ25}.


Early work on learning stabilizer states was initiated by Aaronson and Gottesman~\cite{AG08}, who showed that stabilizer states are efficiently learnable. The first complete proof that $O(n)$ copies suffice was given by Montanaro~\cite{Mon17}, via a procedure called Bell difference sampling that runs in $O(n^3)$ time. For the agnostic learning of qubit stabilizer states, Grewal et~al.~\cite{GIKL24b} first gave an algorithm for pure states with polynomial sample complexity but exponential time complexity. In the regime where the fidelity with the nearest stabilizer state exceeds $\cos^2(\pi/8)$, they further provided a polynomial-time algorithm. Subsequently, Chen et~al.~\cite{CGYZ25} proposed a stabilizer bootstrapping framework that solves the general agnostic learning problem for mixed states with quasi-polynomial sample and time complexity. Meanwhile, learning qudit stabilizer states in the realizable setting was studied by Allcock et~al.~\cite{ADIS25}, who gave an efficient algorithm with polynomial sample and time complexity. However, the agnostic learning of qudit stabilizer states remains largely unexplored. This naturally raises the following question:

\vspace{1em}
{\centering \it Does there exist an efficient quantum algorithm for agnostic learning of stabilizer states on qudit systems?}
\vspace{1em}

In this paper, we answer this question affirmatively for $n$-qudit pure states when the dimension $d>2$ is a prime.


\subsection{Our results}

We present the first efficient quantum algorithm for pure-state agnostic learning of stabilizer states on qudit systems. The goal is to find a stabilizer state that approximates an unknown state  as closely as possible. To state our results precisely, we adopt the stabilizer fidelity~\cite{BBC19} as the figure of merit, defined as
\begin{equation}
    F_{\mathcal{S}}(\ket{\psi}) = \max_{\ket{\phi}\in \mathcal{S}_d^n} 
    \left|\braket{\phi|\psi}\right|^2,
\end{equation}
where $\mathcal{S}_d^n$ denotes the set of all stabilizer states in 
$(\mathbb{C}^d)^{\otimes n}$. Our main result is the following theorem.


\begin{theorem}[Informal]\label{theorem 1.1}
Let $d>2$ be a prime and $\varepsilon,\tau>0$. There exists a quantum algorithm that, given copies of an unknown $n$-qudit pure state $\ket{\psi} \in (\mathbb{C}^d)^{\otimes n}$ with $F_{\mathcal{S}}(\ket{\psi}) \geq \tau$, outputs a classical description of a stabilizer state $\ket{\phi} \in (\mathbb{C}^d)^{\otimes n}$ such that $\left| \braket{\phi|\psi} \right|^2 \geq F_{\mathcal{S}}(\ket{\psi})-\varepsilon$ with high probability, using $O\left(n \log n \right) \cdot \left(\frac{d}{\tau}\right)^{O(d^2 \log \frac{1}{\tau})}+O\left( \frac{d}{\varepsilon^2} \log^2\frac{1}{\tau}  \right)$ samples and  $O\left(\frac{n^3}{\varepsilon^2} \right) \cdot \left(\frac{d}{\tau}\right)^{O(d^2 \log \frac{1}{\tau})}$ time.
\end{theorem}

Our algorithm remains efficient when $\tau = \Omega(1)$, achieving polynomial complexity comparable to that of the learning algorithms on qudit systems in the realizable setting~\cite{ADIS25}. Notably, the complexity becomes quasi-polynomial when $\tau = 1/\mathrm{poly}(n)$, mirroring the quasi-polynomial behavior established for the agnostic setting on qubits~\cite{CGYZ25}. A detailed comparison of our results with several known results in agnostic learning of stabilizer states is provided in Table~\ref{table}.


\begin{table}[!t]
\footnotesize
\centering
\caption{The comparison of our results with several known results in agnostic learning of stabilizer states.}
\label{table}
\tabcolsep 3pt 
\renewcommand{\arraystretch}{1.5}
\resizebox{\textwidth}{!}{
\begin{tabular*}{\textwidth}{|c|c|c|c|c|c|}
\toprule
Work & Target states & Dimension &Stabilizer fidelity  &Sample complexity &Time complexity \\\hline 
 \multirow{2}{*}{Grewal et~al.~\cite{GIKL24b}} & \multirow{2}*{Pure} & \multirow{2}*{$d=2$} & $\tau >0$  & $O\left(\frac{n}{\varepsilon^2\tau^4}\right)$  & $\frac{1}{\varepsilon^2} \cdot 3^{O\left( \frac{n}{\tau^4} \right)}$\\[10pt]
  \cline{4-6} ~ & ~ & ~ & \makecell{$\tau \geq \cos^2\left( \frac{\pi}{8} \right)+\gamma$}  & \makecell{$O\left(n+\frac{\log n}{\gamma^2}\right)$}  & \makecell{$O\left( n^3+\frac{n^2\log n}{\gamma^2} \right)$}\\[10pt]\hline
  Chen et~al.~\cite{CGYZ25} & Mixed & $d=2$ & $\tau>0$ &  $O\left( n \right) \cdot \left(\frac{1}{\tau}\right)^{O(\log \frac{1}{\tau})}+O\left( \frac{1}{\varepsilon^2} \log^2\frac{1}{\tau}  \right)$  & $O\left(n^2\left( n+\frac{1}{\varepsilon^2} \right) \right) \cdot \left(\frac{1}{\tau}\right)^{O(\log \frac{1}{\tau})} $ \\[10pt]\hline 
 \multirow{2}*{\textbf{Ours}} & \multirow{2}*{Pure} & \multirow{2}*{$d>2$ prime} &$\tau>0$ & $O\left(n \log n \right) \cdot \left(\frac{d}{\tau}\right)^{O(d^2 \log \frac{1}{\tau})}+O\left( \frac{d}{\varepsilon^2} \log^2\frac{1}{\tau}  \right)$   & $O\left(\frac{n^3}{\varepsilon^2} \right) \cdot \left(\frac{d}{\tau}\right)^{O(d^2 \log \frac{1}{\tau})}$  \\ [10pt]
 \cline{4-6}~ & ~ & ~ & \makecell{$\tau \geq \cos^2\left( \frac{\pi}{8} \right)+\gamma$}  & \makecell{$O\left( \frac{d^2}{\gamma^2} n \log n \right)$}   & \makecell{$O\left( n^3+\frac{n^2d^2}{\gamma^2}\log n \right)$}\\

\bottomrule
\end{tabular*}
}
\end{table}

A compelling application of our main result is the efficient estimation  of magic, which quantifies the degree to which a quantum state departs from the stabilizer formalism. A natural measure of magic is the stabilizer fidelity $F_{\mathcal{S}}(\ket{\psi})$. 
While Chen et~al.~\cite{CGYZ25} recently achieved quasi-polynomial complexity for estimating this quantity on qubit systems, the qudit setting has remained out of reach. Our main result directly implies the following theorem.

\begin{theorem}[Informal]\label{theorem 1.2}
Let $d>2$ be a prime and $\varepsilon>0$. There exists a quantum algorithm that, given copies of an unknown $n$-qudit pure state $\ket{\psi} \in (\mathbb{C}^d)^{\otimes n}$, estimates $F_{\mathcal{S}}(\ket{\psi})$ to within additive error $\varepsilon$ with high probability, using $O\left(n \log n \right) \cdot \left(\frac{d}{\varepsilon}\right)^{O(d^2 \log \frac{1}{\varepsilon})}$ samples and $O\left(n^3 \right) \cdot \left(\frac{d}{\varepsilon}\right)^{O(d^2 \log \frac{1}{\varepsilon})}$ time.
\end{theorem}

Beyond the general quasi-polynomial algorithm of Theorem~\ref{theorem 1.1}, the high-fidelity regime $F_{\mathcal{S}}(|\psi\rangle) > \cos^2(\pi/8)$ admits a qualitatively simpler treatment. A key structural fact, established in~\cite{GIKL24b}, is that in this regime the squared magnitudes of the Weyl coefficients take strictly separated values depending on whether the operator lies inside or outside the stabilizer group, enabling a direct threshold-based approach that bypasses the stabilizer bootstrapping machinery. Generalizing this approach to qudits is non-trivial. The analogous qubit argument exploits the Hermiticity of Weyl operators to invoke an uncertainty relation, a route that is unavailable in the qudit setting where Weyl operators are non-Hermitian. We prove the same lemma via an independent algebraic argument, yielding a streamlined polynomial-time algorithm for agnostic learning in the high-fidelity regime.

\begin{theorem}[Informal]\label{theorem 1.3}
Let $d>2$ be a prime and $\gamma > 0$. There is a quantum algorithm that, given copies of an $n$-qudit pure state $\ket{\psi} \in (\mathbb{C}^d)^{\otimes n}$ that has fidelity at least $\cos^2(\pi/8) + \gamma$ with some stabilizer state $\ket{\phi}$, outputs $\ket{\phi}$ with high probability, using $O\left( \frac{d^4}{\gamma^2} n \log n \right)$ samples and $O\left( n^3+\frac{n^2d^4}{\gamma^2}\log n \right)$ time.
\end{theorem}

\subsection{Our Techniques}

To achieve the results in Theorem~\ref{theorem 1.1}, our algorithm builds upon the \emph{stabilizer bootstrapping} framework recently introduced by Chen et~al.~\cite{CGYZ25}. We illustrate the framework here via its application to qubit stabilizer states for concreteness. At a high level, it proceeds as follows. Let $\ket{\phi}$ denote the stabilizer state achieving maximum fidelity with the input state $\ket{\psi}$, and suppose $|\langle \phi | \psi \rangle|^2 \ge \tau$. The goal is to identify a complete set of generators for the (unsigned) stabilizer group of $\ket{\phi}$, denoted $\mathrm{Weyl}(\ket{\phi})$; once these generators are in hand, measuring $\ket{\psi}$ in their joint eigenbasis collapses the state to $\ket{\phi}$ with probability at least $\tau$.

The primary algorithmic primitive is \emph{Bell difference sampling}~\cite{Mon17,GNW21}. In the qubit setting, for $\mathbf{x}=(\mathbf{a},\mathbf{b}) \in \mathbb{F}_2^{2n}$ with $\mathbf{a}=(a_1, \ldots, a_n)$ and $\mathbf{b} = (b_1, \ldots, b_n)$, the \emph{Weyl operator} $W_{\mathbf{x}}$ is defined as
\begin{equation}
    W_{\mathbf{x}} = \mathrm{i}^{\mathbf{a} \cdot \mathbf{b}}X^{a_1}Z^{b_1} \otimes \cdots \otimes X^{a_n}Z^{b_n},
\end{equation}
where $\mathbf{a} \cdot \mathbf{b} = \sum_{j=1}^n{a_jb_j} \in \mathbb{Z}$. The general qudit analogue is defined in Eq.~\eqref{Weyl_operators}. The squared correlations $\bigl|\bra{\psi}W_{\mathbf{x}}\ket{\psi}\bigr|^2$ define a probability distribution $p_\psi(\mathbf{x}) = 2^{-n}\bigl|\bra{\psi}W_{\mathbf{x}}\ket{\psi}\bigr|^2$ over $\mathbb{F}_2^{2n}$. Performing Bell difference sampling on four copies of $\ket{\psi}$ is equivalent to drawing from the self-convolution
\begin{equation}\label{Weyl_distribution_qubit}
    q_\psi(\mathbf{x}) = \sum_{\mathbf{y} \in \mathbb{F}_2^{2n}}{p_\psi(\mathbf{y})p_\psi(\mathbf{x}+\mathbf{y})},
\end{equation}
known as the \emph{Weyl distribution} of $\ket{\psi}$. The utility of this primitive rests on a crucial anti-concentration property established in prior work~\cite{GIKL24b, CGYZ25}: the Weyl distribution $q_\psi$ places mass at least $\Omega(\tau^4)$ on $\mathrm{Weyl}(\ket{\phi})$, and this mass is sufficiently spread that it cannot concentrate on any proper subgroup. 

This anti-concentration property powers an iterative recovery procedure that progressively refines a working state $\ket{\psi^\prime}$, initially set to $\ket{\psi}$. In each round, one draws multiple samples from the Weyl distribution $q_{\psi^\prime}$ and retains all sampled Weyl operators $W_{\mathbf{x}}$ satisfying $\bigl|\bra{\psi^\prime}W_{\mathbf{x}}\ket{\psi^\prime}\bigr|^2 > 1/2$; such operators are guaranteed to commute with each other. If the retained high-correlation operators generate the complete stabilizer group $\mathrm{Weyl}(\ket{\phi})$, the algorithm proceeds to the final measurement step. Otherwise, the anti-concentration ensures that $q_{\psi^\prime}$ places appreciable mass on stabilizers of $\ket{\phi}$ outside the subgroup generated by the current high-correlation set, and consequently one obtains a low-correlation stabilizer with probability $\Omega(\tau^4)$. Measuring $\ket{\psi^\prime}$ with this low-correlation stabilizer and post-selecting on a randomly guessed eigenspace yields, when the guess is correct, a post-selected state whose fidelity with $\ket{\phi}$ is amplified by a constant factor, and then the working state $\ket{\psi^\prime}$ is updated to this post-selected state. Since each round amplifies the fidelity of $\ket{\psi^\prime}$ with $\ket{\phi}$ by a constant factor, and fidelity is bounded above by $1$, at most $O(\log(1/\tau))$ successful rounds suffice before the high-correlation operators generate the complete stabilizer group. Ultimately, the entire iterative procedure outputs a stabilizer state with high probability.

However, generalizing the stabilizer bootstrapping algorithm to qudit systems is not straightforward and requires addressing the following three key challenges.

\paragraph{Failure of standard Bell difference sampling} 

Recall that in the qubit setting, Bell difference sampling produces samples from the Weyl distribution $q_\psi$ of the input state $\ket{\psi}$ as in Eq.~\eqref{Weyl_distribution_qubit}. When one directly applies this protocol to qudits, however, the resulting distribution takes a qualitatively different form~\cite{ADIS24}:
\begin{equation}
    q_\psi(\mathbf{x}) = \sum_{\mathbf{y} \in \mathbb{F}_d^{2n}}{p_\psi(\mathbf{y})p_\psi(J(\mathbf{x}-\mathbf{y}))},
\end{equation}
where $p_\psi(\mathbf{y}) = d^{-n}\bigl|\bra{\psi}W_{\mathbf{y}}\ket{\psi}\bigr|^2$, and $J\colon (\mathbf{u},\mathbf{v}) \mapsto (-\mathbf{u},\mathbf{v})$ for all $(\mathbf{u},\mathbf{v}) \in \mathbb{F}_d^{2n}$. Over $\mathbb{F}_2$, this map is the identity, and this definition coincides with Eq.~\eqref{Weyl_distribution_qubit}.  For $d > 2$, however, $J$ is a nontrivial involution that fundamentally distorts the output distribution. Even when the input is a stabilizer state~$\ket{\phi}$, the distribution can be supported on the entire ambient space $\mathbb{F}_d^{2n}$, rendering the samples completely uninformative. This distortion equally undermines the algebraic structure exploited in the anti-concentration arguments central to agnostic learning.

To address this obstacle, we employ the \emph{skewed Bell difference sampling} protocol of~\cite{ADIS25} in the odd prime setting. The procedure applies an entangling unitary, constructed by exploiting the existence of $a_1, a_2 \in \mathbb{F}_d$ satisfying $a_1^2 + a_2^2 \equiv -1 \pmod{d}$, to four of eight copies of the input state, then performs standard Bell measurements on each of the four resulting pairs and takes pairwise differences of the outcomes, producing two strings $\mathbf{x}_1, \mathbf{x}_2 \in \mathbb{F}_d^{2n}$. When the input is a stabilizer state~$\ket{\phi}$, each string is independently and uniformly distributed over $\mathrm{Weyl}(\ket{\phi})$.

However, the agnostic setting requires understanding the joint distribution of $(\mathbf{x}_1, \mathbf{x}_2)$ for an arbitrary input state~$\ket{\psi}$, where the two strings are no longer independent and the distribution is significantly more complex to analyze. We sidestep this difficulty by working with the marginal distribution $\mathbf{B}_\psi$ of a single output string, retaining only one string from each run of the protocol. Establishing the requisite anti-concentration property of $\mathbf{B}_\psi$ is the subject of the second challenge.

\paragraph{Anti-concentration properties of the skewed Bell difference sampling}

Let $d>2$ be a prime, and $\ket{\psi} \in (\mathbb{C}^d)^{\otimes n}$ an $n$-qudit state. Let $\ket{\phi}$ denote the stabilizer state achieving maximum fidelity $F$ with $\ket{\psi}$, and write $S = \mathrm{Weyl}(\ket{\phi})$ for its unsigned stabilizer group. We prove the following anti-concentration bound for the marginal distribution $\mathbf{B}_\psi$, namely for any proper subspace $T \subsetneq S$,
\begin{equation}
    \sum_{\mathbf{x} \in S \setminus T} \mathbf{B}_\psi(\mathbf{x}) \geq \frac{(d-1)\,F^6}{4d^3}.
\end{equation}
Our proof follows the high-level structure of the qubit argument in~\cite{GIKL24b}. Specifically, the total mass $\mathbf{B}_{\psi}(S \setminus T)$ is lower-bounded by the product of the total mass that $p_{\psi}$ places on $T$ and the sum of characteristic function values over $S \setminus T$. The first factor can be bounded from below using standard techniques, and the second requires a genuinely new argument.

In the qubit proof, this second factor is bounded by fixing all but one qubit and reducing it to a single-variable optimization that can be solved explicitly. The analogous reduction for qudits yields an optimization over $d - 1$ complex variables with $O(d^2)$ constraints, which does not admit a clean closed-form solution.

Our key insight is that this optimization can be bypassed entirely. By fixing all but one qudit and invoking the global optimality of $\ket{\phi}$, we extract a normalized single-qudit state $\ket{\psi_1}$ whose stabilizer fidelity $F_{\mathcal{S}}(\ket{\psi_1})$ is attained by a single-qudit stabilizer state induced by $\ket{\phi}$. It then suffices to establish a universal lower bound on $F_{\mathcal{S}}(\ket{\psi_1})$ for arbitrary single-qudit states. We show that $F_{\mathcal{S}}(\ket{\psi_1}) \geq 2/(d+1)$ by exploiting the fact that the $d(d+1)$ single-qudit stabilizer states not only constitute $d+1$ mutually unbiased bases, but also form a complex projective $2$-design~\cite{KR05}. 
Combining this bound with the optimality of $\ket{\phi}$ yields the required result, completing the anti-concentration argument.

\paragraph{Non-Hermiticity of Weyl operators} 

In the qubit setting, Weyl operators are Hermitian, a property exploited by prior work in ways that break down for qudits.

Within the bootstrapping framework, the algorithm must estimate the squared correlations $\bigl|\bra{\psi}W_\mathbf{x}\ket{\psi}\bigr|^2$ in order to classify each sampled operator as having either high or low correlation. For qubit Weyl operators, each $W_\mathbf{x}$ is a $\pm 1$-valued Hermitian observable, so $\bra{\psi}W_\mathbf{x}\ket{\psi}$ is real-valued and can be estimated via Bernoulli mean estimation. This enables an efficient batch estimation procedure~\cite{HKP21}, in which all $m$ squared correlations are simultaneously estimated to accuracy $\varepsilon$ using only $O(\log(m)/\varepsilon^2)$ copies. For qudit Weyl operators, $\bra{\psi}W_\mathbf{x}\ket{\psi}$ is generically complex, and this Bernoulli reduction no longer applies. We instead estimate each squared correlation individually using the SWAP test~\cite{BCW+01}, at a cost of $O(m\log{(m)}/\varepsilon^2)$ copies for $m$ operators.


A similar issue arises in the high-fidelity regime discussed in Theorem~\ref{theorem 1.3}. The qubit proof relies on an uncertainty relation that requires Hermiticity, but we successfully bypass this obstacle via an independent algebraic argument.

\subsection{Related works}

Bell difference sampling is a fundamental tool for extracting information about stabilizer states. It was introduced by Montanaro~\cite{Mon17} for identifying an unknown stabilizer state and subsequently formalized and studied more thoroughly by Gross et~al.~\cite{GNW21}. The technique has since been extended in the development of algorithms for stabilizer states and states close to stabilizer states~\cite{LC22,GIKL23,SM25}. However, Allcock et~al.~\cite{ADIS24} showed that a natural extension of Bell difference sampling to qudits, namely measuring two copies of a given $n$-qudit state in the generalized Bell basis, can provide no information about the corresponding stabilizer group in the worst case. Instead, they proposed an algorithm based on techniques from hidden polynomial problems to learn $n$-qudit stabilizer states of prime dimension using $O(n)$ samples and $O(n^4)$ time in the realizable setting. Later, an algorithm based on general Fourier sampling~\cite{HEC25} was applied to all dimensions, requiring $O(nd)$ copies and $O(n^3 + nd)$ time. Recently, Allcock et~al.~\cite{ADIS25} successfully extended Bell difference sampling to qudits of all dimensions while preserving the structural properties that make it useful in the qubit case. By employing this protocol, dubbed skewed Bell difference sampling, they designed a quantum algorithm that identifies an unknown stabilizer state for all dimensions using $O(n)$ copies and $O(n^3)$ time.\footnote{More precisely, the copy and time complexities are $O(n + \log(\ell/\delta))$ and $O(n^3 + n^2\log(\ell/\delta))$, respectively, where $d = \prod_{i=1}^{\ell} p_i^{k_i}$ is the prime factorization of $d$ and $\delta$ is the failure probability.}

Significant progress has been made on the agnostic learning of quantum states, a quantum analogue of the classical agnostic learning framework~\cite{Val84, KKS92}. In 2024, Grewal et~al.~\cite{GIKL24b} proposed a quantum algorithm with polynomial sample complexity but exponential time complexity for finding a qubit stabilizer state whose fidelity with an unknown pure state is within $\varepsilon$ of optimal. They also investigated how to learn a stabilizer product state to approximate an unknown mixed state~\cite{GIKL24a}, achieving quasi-polynomial sample and time complexities. Bakshi et~al.~\cite{BBK25} studied the problem of finding a product state to approximate an unknown mixed state and developed a learning algorithm using $N = n^{\mathrm{poly}(1/\varepsilon)}$ copies and $\mathrm{poly}(N)$ classical overhead. Chen et~al.~\cite{CGYZ25} introduced the stabilizer bootstrapping framework, achieving agnostic learning of qubit stabilizer states for unknown mixed states in quasi-polynomial time $(1/\tau)^{O(\log(1/\tau))} \cdot \mathrm{poly}(n, 1/\varepsilon)$, a significant improvement over the exponential time complexity of~\cite{GIKL24b}. The framework also extends to other concept classes, including stabilizer product states, discrete product states, and states with high stabilizer dimension. The study of quantum agnostic learning of phase states was recently initiated in~\cite{ADGO26}.

\subsection{Organization}

The remainder of this paper is organized as follows. Section~\ref{preliminaries} introduces the necessary preliminaries, including the symplectic vector space over $\mathbb{F}_d$, the generalized Pauli group, stabilizer states, skewed Bell difference sampling, and three standard algorithmic primitives used throughout the paper. Section~\ref{sampling property} establishes the anti-concentration property of skewed Bell difference sampling, which is the key technical ingredient of our approach. Section~\ref{bootstrapping} presents the stabilizer bootstrapping algorithm for qudit systems and proves Theorems~\ref{theorem 1.1} and~\ref{theorem 1.2}. Section~\ref{bounded-distance section} gives a polynomial-time algorithm for the high-fidelity regime, proving Theorem~\ref{theorem 1.3}. We conclude with a discussion of further investigation in Section~\ref{conclusion}.

\section{Preliminaries}\label{preliminaries}

In this section, we introduce the fundamental concepts and elementary results required throughout this paper. Let $d >2$ be a prime and denote by $\mathbb{F}_{d} = \{0,1,\ldots,d-1\}$ the finite field of integers modulo $d$. Let $\omega = \mathrm{e}^{2\pi \mathrm{i}/d}$ be a primitive $d$-th root of unity and set $\kappa = (-1)^d \mathrm{e}^{\mathrm{i}\pi/d} = \mathrm{e}^{\mathrm{i}\pi(d^2+1)/d}$. One can verify that $\kappa^2 = \omega$ and $\kappa^d = 1$. For the sake of brevity, we shall write $\psi$ to stand for the pure‑state density operator $\ket{\psi}\bra{\psi}$, while $I$ will consistently denote the identity operator (or identity matrix, according to the context). Finally, for any two quantum pure states $\ket{\psi}$ and $\ket{\phi}$, the fidelity between them is defined by $F(\ket{\psi},\ket{\phi}) = \left| \braket{\psi|\phi} \right|^2$. As far as computational complexity is concerned, we assume that measuring a qudit in the computational basis, as well as applying single-qudit and two-qudit quantum gates drawn from a universal gate set~\cite{MS00,BOB05,BBO06}, each requires $O(1)$ time to execute. 

In addition, our analysis will rely on the following Chernoff bound and Hoeffding’s inequality.
\begin{fact}[Chernoff bound]
Let $X_1,X_2,\dots,X_n$ be independent identically distributed random variables taking values in $\{0,1\}$. Let $X=\sum_{i=1}^n X_i$ and $\mu=\mathbf{E}[X]$. Then for any $\delta>0$,
\begin{equation}
\Pr[X\leq (1-\delta)\mu] \leq \mathrm{e}^{-\delta^2\mu/2}.
\end{equation}
\end{fact}

\begin{fact}[Hoeffding’s inequality]
Let $X_1,X_2,\dots,X_n$ be independent random variables subject to $a_i \leq X_i \leq b_i$ for all $i$. Let $X=\sum_{i=1}^n X_i$ and $\mu=\mathbf{E}[X]$. Then for all $t>0$ it holds that
\begin{equation}
\Pr[X-\mu \geq t] \leq \exp\left( -\frac{2t^2}{\sum_{i=1}^n(b_i-a_i)^2} \right)
\end{equation}
and
\begin{equation}
\Pr[|X-\mu| \geq t] \leq 2\exp\left( -\frac{2t^2}{\sum_{i=1}^n(b_i-a_i)^2} \right).
\end{equation}
\end{fact}

\subsection{Symplectic vector spaces over $\mathbb{F}_d$}

In this work, $\mathbb{F}_d^{2n}$ is equipped with a canonical symplectic bilinear form, which we call the symplectic product.
\begin{definition}[Symplectic product]
For $\mathbf{x},\mathbf{y}\in\mathbb{F}_d^{2n}$, we define the symplectic product as 
\begin{equation}
[\mathbf{x},\mathbf{y}]= \sum_{i=1}^n (x_iy_{n+i}-x_{n+i}y_i) \pmod d.
\end{equation}
\end{definition}

In analogy with the orthogonal complement of the standard inner product, the symplectic product naturally gives rise to the notion of a symplectic complement.
\begin{definition}[Symplectic complement]
Let $A \subseteq \mathbb{F}_d^{2n}$ be a subspace. The symplectic complement of $A$, written as $A^\perp$, is defined by 
\begin{equation}
A^\perp=\{ \mathbf{x} \in \mathbb{F}_d^{2n} : \forall\mathbf{a}\in A, [\mathbf{x},\mathbf{a}]=0\}.
\end{equation}
\end{definition}

In the following, we summarize several key properties of the symplectic complement, which are analogous to those of the more familiar orthogonal complement.
\begin{fact}
Let $A$ and $B$ be subspaces of $\mathbb{F}_d^{2n}$. Then,
\begin{enumerate}
\item[(a)] $A^\perp$ is a subspace.
\item[(b)] $(A^\perp)^\perp=A$.
\item[(c)] $|A|\cdot |A^\perp|=d^{2n}$, or equivalently $\dim(A)+\dim(A^\perp)=2n$.
\item[(d)] $A \subseteq B \Longrightarrow B^\perp \subseteq A^\perp$.
\item[(e)] $(A+B)^\perp = A^\perp \cap B^\perp$ where $A+B=\{\mathbf{a}+\mathbf{b}: \mathbf{a}\in A, \mathbf{b} \in B\}$.
\end{enumerate}
\end{fact}

Let $A\subseteq \mathbb{F}_d^{2n}$ be a subspace. If $[\mathbf{x},\mathbf{y}]=0$ for all $\mathbf{x},\mathbf{y} \in A$, we call $A$ \emph{isotropic}. Moreover, for any $\mathbf{z} \in \mathbb{F}_d^{2n}$, the following standard identity holds:
\begin{equation}
    \sum_{\mathbf{x} \in A} \omega^{[\mathbf{x},\mathbf{z}]} = 
\begin{cases}
|A|, & \text{if } \mathbf{z} \in A^\perp; \\
0, & \text{otherwise}.
\end{cases}
\end{equation}
A subspace $T\subseteq \mathbb{F}_d^{2n}$ is \emph{Lagrangian} if $T^\perp = T$. Equivalently, Lagrangian subspaces can be characterized as isotropic subspaces of maximal dimension $n$.

\subsection{Generalized Pauli group and Weyl operators}

Define the unitary shift and clock operators $X$ and $Z$, respectively, as
\begin{equation}
X\ket{q}=\ket{q+1},~Z\ket{q}=\omega^q\ket{q}, ~\forall q\in \mathbb{F}_d.
\end{equation}
The Pauli group on a single qudit, denoted $\mathcal{P}_d$, is generated by the clock and shift operators with integer powers of the phase factor $\kappa$, i.e., $\mathcal{P}_d = \langle \kappa I, X, Z \rangle$. For an $n$-qudit system over $\mathbb{F}_d$, the corresponding Pauli group is simply the $n$-fold tensor product $\mathcal{P}_d^n = \langle \kappa I, X, Z \rangle^{\otimes n}$. Any element of $\mathcal{P}_d^n$ is referred to as a Pauli gate. The associated $n$-qudit Clifford group $\mathcal{C}_d^n$ consists of all unitary gates $C$ that normalize the Pauli group in the sense that $C P C^\dagger \in \mathcal{P}_d^n$ for every $P \in \mathcal{P}_d^n$.

The qudit Clifford group is generated by three elementary gates $\{\mathcal{F}, S', \mathrm{SUM}\}$, which generalize the Hadamard, phase, and CNOT gates in the qubit setting. Their actions on the computational basis are defined as follows: for all $i,j\in \mathbb{F}_d$,
\begin{equation}
\mathcal{F}\ket{i}
   = \frac{1}{\sqrt{d}}\sum_{k=0}^{d-1}\omega^{ik}\ket{k},~ S'\ket{i} = \omega^{i(i-1)/2}\ket{i},~ \mathrm{SUM}\bigl(\ket{i}\otimes\ket{j}\bigr)
   = \ket{i}\ket{(i+j) \bmod d}.
\end{equation}
Their conjugation actions on the Pauli operators $X$ and $Z$
are given by
\begin{align}\label{eq. conjugation}
& \mathcal{F}X\mathcal{F}^\dagger=Z,~
  \mathcal{F}Z\mathcal{F}^\dagger=X^{-1}; \quad
S'XS'^\dagger=XZ,~ S'ZS'^\dagger=Z; \nonumber\\
& \mathrm{SUM}(X\otimes I)\mathrm{SUM}^\dagger=X\otimes X,~
 \mathrm{SUM}(I\otimes X)\mathrm{SUM}^\dagger=I\otimes X,\\
& \mathrm{SUM}(Z\otimes I)\mathrm{SUM}^\dagger=Z\otimes I,~
 \mathrm{SUM}(I\otimes Z)\mathrm{SUM}^\dagger
  =Z^{-1}\otimes Z. \nonumber
\end{align}

For $\mathbf{x}=(\mathbf{a},\mathbf{b})\in \mathbb{F}_d^{2n}$, the Weyl operator $W_{\mathbf{x}}$ is defined as 
\begin{equation}\label{Weyl_operators}
W_{\mathbf{x}}=\kappa^{\mathbf{a} \cdot \mathbf{b}}(X^{a_1}Z^{b_1}) \otimes \cdots \otimes (X^{a_n}Z^{b_n}).
\end{equation}
Henceforth, we refer to any vector $\mathbf{x} \in \mathbb{F}_d^{2n}$ as a \textit{Pauli string}. For convenience, we shall occasionally use Pauli strings and their corresponding Weyl operators interchangeably.

The action of a Weyl operator $W_{\mathbf{x}}$ on the Hilbert space $(\mathbb{C}^d)^{\otimes n}$ is given explicitly by
\begin{equation}
W_{\mathbf{x}}\ket{\mathbf{q}}=\kappa^{\mathbf{a} \cdot \mathbf{b}} \omega^{\mathbf{b} \cdot \mathbf{q}} \ket{(\mathbf{q}+\mathbf{a})\pmod d},~\forall \mathbf{q}\in\mathbb{F}_d^{n}.
\end{equation}
It is clear that every Weyl operator is a Pauli operator, and conversely, any Pauli operator coincides with a Weyl operator up to an overall phase factor given by a power of $\kappa$. More importantly, the Weyl operators satisfy the following fundamental properties.
\begin{property}
For $\mathbf{x},\mathbf{y}\in \mathbb{F}_d^{2n}$,
\begin{enumerate}
\item[(a)] $W_{\mathbf{x}}W_{\mathbf{y}}=\kappa^{[\mathbf{x},\mathbf{y}]}W_{\mathbf{x}+\mathbf{y}}=\omega^{[\mathbf{x},\mathbf{y}]}W_{\mathbf{y}}W_{\mathbf{x}}$.
\item[(b)] $W_{\mathbf{x}}$ and $W_{\mathbf{y}}$ commute if and only if $[\mathbf{x},\mathbf{y}]=0$.
\item[(c)] $W_{\mathbf{x}}^m=W_{m\mathbf{x}}$ for all $m\in \mathbb{F}_d$. In particular, $W_{\mathbf{x}}^\dagger=W_{-\mathbf{x}}$ and $W_{\mathbf{x}}^d=I$.
\end{enumerate}
\end{property}

Since the Clifford group normalizes the Pauli group, any Clifford circuit induces a well‑defined action on the phase space $\mathbb{F}_d^{2n}$ via the conjugation of the associated Weyl operators. Concretely, for any Clifford circuit $C$ and any vector $\mathbf{x} \in \mathbb{F}_d^{2n}$, there exists a unique phase exponent $s \in \mathbb{F}_d$ and a unique vector $\mathbf{y} \in \mathbb{F}_d^{2n}$ such that $CW_{\mathbf{x}}C^{\dagger}=\omega^s W_{\mathbf{y}}$. This correspondence defines the action of $C$ on $\mathbb{F}_d^{2n}$ by setting $C(\mathbf{x}) = \mathbf{y}$. For an arbitrary subset $S \subseteq \mathbb{F}_d^{2n}$, we write $C(S) = \{ C(\mathbf{x}) : \mathbf{x} \in S \}$. Observe, for instance, that $C(\mathbb{F}_d^{2n}) = \mathbb{F}_d^{2n}$.

\subsection{Stabilizer state and Weyl expansion}

A pure state $\ket{\psi} \in (\mathbb{C}^d)^{\otimes n}$ is called a \textit{stabilizer state} if it is stabilized by $d^n$ mutually commuting Pauli operators $P \in \mathcal{P}_d^n$; that is, if $P \ket{\psi} = \ket{\psi}$ for all such $P$. We denote the set of all stabilizer states in $(\mathbb{C}^d)^{\otimes n}$ by $\mathcal{S}_d^n$. One natural measure of the ``stabilizer complexity'' of an arbitrary quantum state is provided by its stabilizer fidelity.
\begin{definition}[Stabilizer fidelity]
Let $\ket{\psi}$ be a pure $n$-qudit state. Its \textit{stabilizer fidelity} is defined as $F_{\mathcal{S}}(\ket{\psi})=\max_{\ket{\phi}\in \mathcal{S}_d^n}\left| \braket{\phi|\psi} \right|^2$.
\end{definition}

\begin{definition}[Unsigned stabilizer group]
For a pure state $\ket{\psi} \in (\mathbb{C}^d)^{\otimes n}$, its \textit{unsigned stabilizer group} is defined as $\mathrm{Weyl}(\ket{\psi})=\{\mathbf{x}\in \mathbb{F}_d^{2n} : W_{\mathbf{x}}\ket{\psi}=\omega^s\ket{\psi} \text{ for some } s\in\mathbb{F}_d\}$.
\end{definition}
The cardinality $|\mathrm{Weyl}(\ket{\psi})|$ is referred to as the \textit{stabilizer size} of $\ket{\psi}$. Observe that for any nonzero state $\ket{\psi} \in (\mathbb{C}^d)^{\otimes n}$, the set $\mathrm{Weyl}(\ket{\psi})$ forms an isotropic subspace of $\mathbb{F}_d^{2n}$. Furthermore, if $\ket{\psi}$ is a stabilizer state, its stabilizer size attains the maximal value $d^n$, which is equivalent to the condition that $\mathrm{Weyl}(\ket{\psi})$ is a Lagrangian subspace.

The Weyl operators collectively constitute an orthogonal basis for the space of $d^n \times d^n$ matrices with respect to the Hilbert--Schmidt inner product $\langle A, B \rangle = \operatorname{Tr}(A^\dagger B)$. This fact naturally leads to the \textit{Weyl expansion} of an arbitrary quantum state.
\begin{definition}[Weyl expansion]
For a pure state $\ket{\psi} \in (\mathbb{C}^d)^{\otimes n}$, the associated density operator $\psi := \ket{\psi}\bra{\psi}$ admits a \textit{Weyl expansion} of the form
\begin{equation}
\psi=d^{-n/2}\sum_{\mathbf{x}\in \mathbb{F}_d^{2n}}c_{\psi}(\mathbf{x})W_{\mathbf{x}},
\end{equation}
where the coefficient $c_{\psi}(\mathbf{x})=d^{-n/2}\bra{\psi}W_{\mathbf{x}}^\dagger\ket{\psi}$ is referred to as the \textit{characteristic function} of the state $\psi$.
\end{definition}
We denote the \textit{characteristic distribution} of $\ket{\psi}$ by $p_{\psi}(\mathbf{x})=\left| c_{\psi}(\mathbf{x}) \right|^2=d^{-n}\left| \bra{\psi}W_{\mathbf{x}}\ket{\psi} \right|^2$. Observe that $p_{\psi}(\mathbf{x}) \leq d^{-n}$ holds for all $\mathbf{x} \in \mathbb{F}_d^{2n}$. When $\ket{\psi}$ is a stabilizer state, we have $p_{\psi}(\mathbf{x}) = d^{-n}$ for every $\mathbf{x} \in \mathrm{Weyl}(\ket{\psi})$ and $p_{\psi}(\mathbf{x}) = 0$ otherwise. Moreover, $p_{\psi}$ is indeed a genuine probability distribution, as the following normalization condition confirms:
\begin{equation}
    \sum_{\mathbf{x}\in \mathbb{F}_d^{2n}}p_{\psi}(\mathbf{x})=\mathrm{Tr}(\ket{\psi}\bra{\psi}\ket{\psi}\bra{\psi})=1.
\end{equation}

The properties of $p_{\psi}(\mathbf{x})$ established below will serve as an essential foundation for the proofs and analysis that follow.

\begin{proposition}[\cite{ADIS25}]\label{X to Xperp}
Let $\ket{\psi} \in (\mathbb{C}^d)^{\otimes n}$ be a pure state and $T \subseteq \mathbb{F}_d^{2n}$ a subspace. Then,
\begin{equation}
\sum_{\mathbf{x}\in T}p_{\psi}(\mathbf{x})=\frac{|T|}{d^n}\sum_{\mathbf{x}\in T^\perp}p_{\psi}(\mathbf{x}).
\end{equation}
\end{proposition}

\begin{proposition}[\cite{ADIS25}]\label{sum lowerbound}
Given a pure state $\ket{\psi} \in (\mathbb{C}^d)^{\otimes n}$, let $\ket{\phi} = \operatorname*{arg\,max}_{\ket{\varphi} \in \mathcal{S}_d^n} \left| \braket{\varphi | \psi} \right|^2$ denote a stabilizer state that attains the maximal stabilizer fidelity $F_{\mathcal{S}}(\ket{\psi})$. Then,
\begin{equation}
\sum_{\mathbf{x}\in \mathrm{Weyl}(\ket{\phi})}p_{\psi}(\mathbf{x}) \geq F^2_{\mathcal{S}}(\ket{\psi}).
\end{equation}
\end{proposition}

\begin{lemma}\label{subspace lowerbound}
Given a pure state $\ket{\psi}\in (\mathbb{C}^d)^{\otimes n}$, let $\ket{\phi}=\arg \max_{\ket{\varphi}\in \mathcal{S}_d^n} \left|\braket{\varphi|\psi}\right|^2$, and $T$ a proper subspace of $\mathrm{Weyl}(\ket{\phi})$. Then,
\begin{equation}
\sum_{\mathbf{x}\in T}p_{\psi}(\mathbf{x}) \geq \frac{|T|}{d^n}F^2_{\mathcal{S}}(\ket{\psi}).
\end{equation}
\end{lemma}

\begin{proof}
Since $T\subseteq \mathrm{Weyl}(\ket{\phi})$, it follows immediately that $T^\perp \supseteq \mathrm{Weyl}(\ket{\phi})^\perp = \mathrm{Weyl}(\ket{\phi})$. Combining Proposition~\ref{X to Xperp} with Proposition~\ref{sum lowerbound}, we obtain
\begin{equation}
\sum_{\mathbf{x}\in T}p_{\psi}(\mathbf{x})=\frac{|T|}{d^n} \sum_{\mathbf{x}\in T^\perp}p_{\psi}(\mathbf{x})\geq \frac{|T|}{d^n} \sum_{\mathbf{x}\in \mathrm{Weyl}(\ket{\phi})}p_{\psi}(\mathbf{x}) \geq 
\frac{|T|}{d^n} F^2_{\mathcal{S}}(\ket{\psi}). 
\end{equation}
\end{proof}

\subsection{Skewed Bell difference sampling}

In order to extract classical information from an unknown quantum state, one necessarily relies on a well‑designed measurement scheme. Bell difference sampling has been extensively investigated both for qubit systems and for their qudit generalizations, and it now emerges as a powerful primitive for a variety of quantum state learning tasks. In the present work, we make use of a specific qudit generalization known as \textit{skewed Bell difference sampling}, as recently introduced in~\cite{ADIS25}.

The standard formulation of skewed Bell difference sampling requires measuring eight copies of a qudit state in two rounds and subsequently taking the difference of the two outcomes. Compared with Bell difference sampling on qubits, this increase in the required number of copies stems from Lagrange's four-square theorem, which guarantees that every integer $m \in \mathbb{N}$ can be expressed as a sum of four squares; i.e., there exists non‑negative integers $a_1, a_2, a_3, a_4 \in \mathbb{N}_0$ such that $m = a_1^2 + a_2^2 + a_3^2 + a_4^2$. In the construction of~\cite{ADIS25}, the authors set $m = D - 1$, where $D = d$ if the dimension $d$ is odd, and $D = 2d$ if $d$ is even. Notably, in the present setting where $d>2$ is a prime, the sampling procedure simplifies considerably and the required number of copies reduces to four. This simplification follows from elementary number‑theoretic facts concerning quadratic residues modulo an odd prime $d$:
\begin{itemize}
    \item If $d \equiv 3 \pmod 4$, there exists non‑negative integers $a_1, a_2$ such that $a_1^2 + a_2^2 \equiv -1 \pmod d$;
    \item If $d \equiv 1 \pmod 4$, there exists a non‑negative integer $a_1$ such that $a_1^2 \equiv -1 \pmod d$.
\end{itemize}
In light of this simplification, the skewed Bell difference sampling protocol used throughout this work can be introduced as follows.

\begin{definition}[Skewed Bell difference sampling~\cite{ADIS25}]
Suppose $d>2$ is a prime and let $\ket{\psi} \in (\mathbb{C}^d)^{\otimes n}$ be a pure state. Let $\mathcal{B}_{\mathbf{R}}$ be a unitary transformation acting on $((\mathbb{C}^d)^{\otimes n})^{\otimes 2}$ and defined by $\mathcal{B}_{\mathbf{R}}\ket{\mathbf{Q}}=\ket{\mathbf{QR} (\bmod d)}$, where $\mathbf{Q} \in \mathbb{F}_d^{n \times 2}$ and $\mathbf{R} \in \mathbb{F}_d^{2 \times 2}$ is a matrix satisfying $\mathbf{R}^\top\mathbf{R}=\mathbf{R}\mathbf{R}^\top \equiv -1 \pmod{d} \mathbf{I}$. Define the unitary $\mathcal{U}_{\mathbf{R}}$ on $((\mathbb{C}^d)^{\otimes n})^{\otimes 4}$ as the operator that acts on $\bigotimes_{i \in [4]} \ket{\mathbf{q}_i}$ by applying $\mathcal{B}_{\mathbf{R}}$ to the pair $(\ket{\mathbf{q}_2}, \ket{\mathbf{q}_4})$ and the identity to the remaining qudits. Then, performing skewed Bell difference sampling on $\ket{\psi}^{\otimes 8}$ yields samples drawn from the following distribution.
\begin{equation*}
b_\psi(\mathbf{x}_1,\mathbf{x}_2)  \triangleq \operatorname{Tr}\left[ (\mathcal{U}_{\mathbf{R}} \otimes \mathcal{U}_{\mathbf{R}}) ( \Pi_{\mathbf{x}_1} \otimes \Pi_{\mathbf{x}_2} ) (\mathcal{U}_{\mathbf{R}}^\dagger \otimes \mathcal{U}_{\mathbf{R}}^\dagger) \psi^{\otimes 8} \right]= \sum_{\mathbf{Y} \in \mathbb{F}_d^{2n \times 2}} \prod_{i=1}^2 p_\psi(\mathbf{x}_i + \mathbf{Y}_i) p_\psi((\mathbf{Y}\mathbf{R})_i).
\end{equation*}
Here, $\Pi_{\mathbf{x}_i} =\sum_{\mathbf{y}\in \mathbb{Z}_d^{2n}} \ket{W_{\mathbf{y}}}\bra{W_{\mathbf{y}}} \otimes \ket{W_{\mathbf{x}_i+\mathbf{y}}}\bra{W_{\mathbf{x}_i+\mathbf{y}}}$, where
$\ket{W_{\mathbf{y}}}=(W_{\mathbf{y}}\otimes I)\ket{\Phi^{+}}$ with $\ket{\Phi^{+}}=d^{-\frac n2}\sum_{\mathbf{q}\in \mathbb{F}_d^n}\ket{\mathbf{q}}^{\otimes 2}$, and the matrix $\mathbf{R} = \begin{pmatrix}
a_1 & a_2  \\
a_2 & -a_1 \end{pmatrix}$, where $a_1^2+a_2^2 \equiv -1 \pmod{d}$.
\end{definition}
In addition, $\mathbf{Y}_i$ denotes the $i$-th column of the matrix $\mathbf{Y}$, while $(\mathbf{Y} \mathbf{R})_i$ stands for the $i$-th column of the matrix product $\mathbf{Y} \mathbf{R}$.

Since each iteration of skewed Bell difference sampling produces a pair of Pauli strings, characterizing their joint distribution in the agnostic setting presents considerable technical difficulty. To circumvent this obstacle, we restrict our attention to the first component of the output pair and analyze the corresponding marginal distribution. Concretely, we focus on the marginal 
\begin{equation}
\mathbf{B}_{\psi}(\mathbf{x}_1):= \sum_{\mathbf{x}_2 \in \mathbb{F}_d^{2n}} b_\psi(\mathbf{x}_1,\mathbf{x}_2).
\end{equation}
Before proceeding to the explicit formula for this marginal distribution, we first establish a preliminary result concerning the closure property of the relevant subspaces.

\begin{lemma}\label{T'=T}
Let $a_1, a_2 \in \mathbb{N}_0$ be non‑negative integers satisfying $a_1^2 + a_2^2 \equiv -1 \pmod d$, and $T$ a proper subspace of $\mathbb{F}_d^{2n}$. Define the set $T' = \{ a_1 \mathbf{y}_1 + a_2 \mathbf{y}_2 : \mathbf{y}_1, \mathbf{y}_2 \in T \}$. Then $T'$ is also a proper subspace of $\mathbb{F}_d^{2n}$, and $T' = T$.
\end{lemma}

\begin{proof}
Since $T \subseteq \mathbb{F}_d^{2n}$ is a subspace, it is closed under both addition and scalar multiplication. Consequently, for any $\mathbf{y}_1, \mathbf{y}_2 \in T$, the linear combination $a_1 \mathbf{y}_1 + a_2 \mathbf{y}_2$ also lies in $T$, which immediately yields the inclusion $T' \subseteq T$.

To establish the reverse inclusion, observe that the congruence $a_1^2 + a_2^2 \equiv -1 \pmod{d}$ forces $\gcd(a_1, a_2, d) = 1$. By Bézout's Lemma, there exists integers $c_1, c_2$ such that
\begin{equation}
a_1 c_1 + a_2 c_2 \equiv 1 \pmod{d}.
\end{equation}

Given an arbitrary $\mathbf{x} \in T$, set $\mathbf{y}_1 = c_1 \mathbf{x}$ and $\mathbf{y}_2 = c_2 \mathbf{x}$, both of which belong to $T$ by closure under scalar multiplication. Substituting these into the definition of $T'$, we obtain
\begin{equation}
a_1 \mathbf{y}_1 + a_2 \mathbf{y}_2 = (a_1 c_1 + a_2 c_2) \mathbf{x} \equiv \mathbf{x} \pmod{d},
\end{equation}
which proves that $\mathbf{x} \in T'$. Hence $T \subseteq T'$.

Combining the two relations, we conclude that $T' = T$, and in particular $T'$ is a proper subspace of $\mathbb{F}_d^{2n}$.
\end{proof}

\begin{proposition}\label{marginal distribution}
The marginal probability distribution obtained from $b_\psi(\mathbf{x}_1, \mathbf{x}_2)$ is given explicitly by
\begin{equation}
\mathbf{B}_{\psi}(\mathbf{x_1})=\sum_{\mathbf{Y} \in \mathbb{F}_d^{2n \times 2}} p_\psi(\mathbf{x}_1 + \mathbf{Y}_1) \prod_{i=1}^2 p_\psi((\mathbf{Y}\mathbf{R})_i).
\end{equation}
Moreover, if $\ket{\psi}$ is a stabilizer state, then
\begin{equation}
\mathbf{B}_{\psi}(\mathbf{x})= 
\begin{cases} 
d^{-n}, & \text{if } \mathbf{x} \in \mathrm{Weyl}(\ket{\psi}); \\
0, & \text{otherwise}.
\end{cases}
\end{equation}
\end{proposition}

\begin{proof}
By the closure of $\mathbb{F}_d^{2n}$ under addition, we have $\mathbf{x} + \mathbb{F}_d^{2n} = \mathbb{F}_d^{2n}$ for every $\mathbf{x} \in \mathbb{F}_d^{2n}$. Together with $\sum_{\mathbf{x} \in \mathbb{F}_d^{2n}} p_\psi(\mathbf{x}) = 1$, we immediately obtain
\begin{align}
\mathbf{B}_{\psi}(\mathbf{x_1})
&=\sum_{\mathbf{x}_2 \in \mathbb{F}_d^{2n}} \sum_{\mathbf{Y} \in \mathbb{F}_d^{2n \times 2}} \prod_{i=1}^2 p_\psi(\mathbf{x}_i + \mathbf{Y}_i) p_\psi((\mathbf{Y}\mathbf{R})_i) \nonumber \\
&= \sum_{\mathbf{Y} \in \mathbb{F}_d^{2n \times 2}} \left( \sum_{\mathbf{x}_2 \in \mathbb{F}_d^{2n}} p_\psi(\mathbf{x}_2 + \mathbf{Y}_2) \right) p_\psi(\mathbf{x}_1 + \mathbf{Y}_1) \prod_{i=1}^2 p_\psi((\mathbf{Y}\mathbf{R})_i) \nonumber\\
&= \sum_{\mathbf{Y} \in \mathbb{F}_d^{2n \times 2}} \left( \sum_{\mathbf{x} \in \mathbb{F}_d^{2n}} p_\psi(\mathbf{x}) \right) p_\psi(\mathbf{x}_1 + \mathbf{Y}_1) \prod_{i=1}^2 p_\psi((\mathbf{Y}\mathbf{R})_i)  \\
&= \sum_{\mathbf{Y} \in \mathbb{F}_d^{2n \times 2}} p_\psi(\mathbf{x}_1 + \mathbf{Y}_1) \prod_{i=1}^2 p_\psi((\mathbf{Y}\mathbf{R})_i). \nonumber
\end{align}

If $\ket{\psi}$ is a stabilizer state, we know that $p_{\psi}(\mathbf{x})=d^{-n}$ when $\mathbf{x}\in \mathrm{Weyl}(\ket{\psi})$ and $p_{\psi}(\mathbf{x})=0$ when $\mathbf{x}\notin \mathrm{Weyl}(\ket{\psi})$. This means $p_\psi(\mathbf{x}_1 + \mathbf{Y}_1) \prod_{i=1}^2 p_\psi((\mathbf{Y}\mathbf{R})_i) \neq 0$ if and only if $\mathbf{x}_1 + \mathbf{Y}_1,(\mathbf{Y}\mathbf{R})_1,(\mathbf{Y}\mathbf{R})_2 \in \mathrm{Weyl}(\ket{\psi})$. Since $\mathrm{Weyl}(\ket{\psi})$ is a subspace, by Lemma~\ref{T'=T} this implies that $\mathbf{x}_1, \mathbf{Y}_1, \mathbf{Y}_2 \in \mathrm{Weyl}(\ket{\psi})$ if and only if $p_\psi(\mathbf{x}_1 + \mathbf{Y}_1) \prod_{i=1}^2 p_\psi((\mathbf{Y}\mathbf{R})_i) \neq 0$. Thus,
\begin{equation}
\mathbf{B}_{\psi}(\mathbf{x})= 
\begin{cases} 
d^{-n}, & \text{if } \mathbf{x} \in \mathrm{Weyl}(\ket{\psi}); \\
0, & \text{otherwise}.
\end{cases}
\end{equation}
\end{proof}

\subsection{Standard algorithmic primitives}

This section collects the standard algorithmic primitives that serve as building blocks for our agnostic learning algorithms. We begin with the Clifford circuit synthesis stated in Theorem~\ref{theorem circuit}. This theorem provides an efficient procedure for synthesizing a Clifford circuit that maps any isotropic subspace of $\mathbb{F}_d^{2n}$ to a canonical form. The complete proof is deferred to Appendix~\ref{Clifford circuit synthesis}.

\begin{theorem}[Clifford circuit synthesis]\label{theorem circuit}
Suppose $d>2$ is a prime. Given $m$ vectors that span an $r\leq n$ dimensional isotropic subspace $A \subseteq \mathbb{F}_d^{2n}$, there exists an efficient algorithm that outputs a Clifford circuit $C$ such that $C(A) = 0^{2n-r} \times \mathbb{F}_d^r$. The algorithm runs in time $O(mn \cdot \min(m,n))$, and the circuit $C$ contains $O(rn)$ gates.
\end{theorem}

Theorem~\ref{theorem circuit} not only enables us to efficiently represent and manipulate stabilizer subspaces, but also lays the groundwork for subsequent measurements. Nevertheless, in the forthcoming agnostic learning algorithm, it remains essential to quantify the overlap between an unknown quantum state and specific Weyl operators or stabilizer states in order to identify and select those that are relevant to the learning task.

For an unknown quantum state $\rho$ and a collection of observables $\{O_i\}$, the classical shadows framework~\cite{HKP20,MYZ25} offers a systematic and efficient approach to estimating the expectation values $\operatorname{Tr}(O_i \rho)$. In particular, when each $O_i$ is chosen as the projector onto a stabilizer state, this framework yields direct estimations of the corresponding stabilizer fidelities.

\begin{theorem}[Estimating fidelities via classical shadows~\cite{MYZ25}]\label{classical shadows}
Let $d>2$ be a prime. Given an $n$-qudit quantum state $\rho$ and $M$ stabilizer states $|\phi_1\rangle, \dots, |\phi_M\rangle$,
there is an algorithm that, with probability at least $1-\delta$, estimates $\langle \phi_i | \rho | \phi_i \rangle$
to additive error at most $\varepsilon$ for all $i$.
The algorithm only uses single-copy measurements.
The sample complexity is $O\left( \frac{d}{\varepsilon^2} \log \frac{M}{\delta} \right)$
and the time complexity is $O\left( \frac{dM}{\varepsilon^2} n^3 \log \frac{M}{\delta} \right)$.
\end{theorem}

However, in the setting considered here, the qudit Weyl operator $W_{\mathbf{x}}$ is not Hermitian in general. Consequently, the classical shadows framework is not directly suited to estimating quantities of the form $\operatorname{Tr}(W_{\mathbf{x}} \rho)^2$. Moreover, since we are primarily interested in pure states $\rho = \ket{\psi}\bra{\psi}$, the target quantity reduces to the squared overlap $\bigl| \bra{\psi} W_{\mathbf{x}} \ket{\psi} \bigr|^2$. In this case, the standard SWAP test provides a natural and efficient means of estimation. Explicitly, we first apply the unitary $W_{\mathbf{x}}$ to the state $\ket{\psi}$, producing the transformed state $\ket{\phi} = W_{\mathbf{x}} \ket{\psi}$. The fidelity $F(\ket{\psi}, \ket{\phi})=\bigl| \bra{\psi} W_{\mathbf{x}} \ket{\psi} \bigr|^2$ is then estimated via the SWAP test. The following result addresses the estimation of multiple Weyl operators, with a detailed proof provided in the Appendix~\ref{SWAP test}.

\begin{theorem}[Estimating correlations via the SWAP test]\label{theorem M estimations}
Let $d>2$ be a prime. For any $n$-qudit pure state $\ket{\psi}$ and any collection of $M$ Weyl operators $W_{\mathbf{x}_1}, \dots, W_{\mathbf{x}_M}$, there exists an algorithm that, with probability at least $1 - \delta$, estimates $\bigl| \bra{\psi} W_{\mathbf{x}_i} \ket{\psi} \bigr|^2$ for every $i \in \{1, \dots, M\}$ to an additive error at most $\varepsilon$.  Its sample complexity is $O\bigl( \frac{ M}{\varepsilon^2} \log \frac{M}{\delta} \bigr)$, and the corresponding time complexity is $O\bigl( \frac{ M}{\varepsilon^2} n \log \frac{M}{\delta} \bigr)$.
\end{theorem}

\section{anti-concentration of skewed Bell difference sampling} \label{sampling property}



Recall that $\ket{\psi} \in (\mathbb{C}^d)^{\otimes n}$ is the unknown pure state, where $d>2$ is a prime. Let $\ket{\phi}$ denote a stabilizer state that attains the maximal fidelity with $\ket{\psi}$, i.e., $|\langle \phi|\psi \rangle|^2 = F_{\mathcal{S}}(\ket{\psi})$, and write $S = \mathrm{Weyl}(\ket{\phi})$ for its unsigned stabilizer group. 

The cornerstone of our agnostic learning algorithms is the anti-concentration property of skewed Bell difference sampling, namely that the marginal distribution $\mathbf{B}_\psi$ places sufficient mass on $S \setminus T$ for every proper subspace $T \subsetneq S$, ensuring that each run of the protocol yields a stabilizer element linearly independent of all previously collected ones with appreciable probability.

\begin{theorem}\label{theorem B_psi}
Given an $n$-qudit pure state $\ket{\psi} \in (\mathbb{C}^d)^{\otimes n}$, where $d>2$ is a prime, let $\ket{\phi}$ be a stabilizer state that maximizes the stabilizer fidelity and $S = \mathrm{Weyl}(\ket{\phi})$. Let $T \subsetneq S$ be a proper subspace of $S$. Then
\begin{equation}
\sum_{\mathbf{x} \in S \setminus T} \mathbf{B}_{\psi}(\mathbf{x}) \geq \frac{(d-1)F_{\mathcal{S}}(\ket{\psi})^6}{4d^3}.
\end{equation}
\end{theorem}

The proof proceeds in two steps. In Section~\ref{subsection property 1}, we establish that the stabilizer fidelity of any single-qudit pure state is at least $2/(d+1)$. In Section~\ref{subsection property 2}, we leverage this bound to establish Theorem~\ref{theorem B_psi}.

\subsection{Bounding the stabilizer fidelity of single qudit states}\label{subsection property 1}

First, we provide explicit formulas for the eigenvalues and corresponding eigenvectors of a single-qudit Pauli operator when $d>2$ is a prime. 

\begin{lemma}\label{lemma Pauli eigen}
Let $d>2$ be a prime, and fix arbitrary elements $a, b \in \mathbb{F}_d$. Then the single-qudit Pauli operator $P = X^a Z^b$ (up to an overall phase) possesses eigenvalues $\lambda_s = \omega^{s}$, and the corresponding eigenvectors are given by
\begin{equation}
\ket{\psi_s}=\frac{1}{\sqrt{d}} \sum_{j=0}^{d-1}\omega^{\left( \frac{ab}{2}(a^{-1}j-1)-s \right)a^{-1}j}\ket{j},
\end{equation}
where $\omega = \mathrm{e}^{2\pi \mathrm{i}/d}$ denotes a $d$-th root of unity, and $s = 0, 1, \dots, d-1$.
\end{lemma} 

\begin{proof}
Suppose $\ket{\psi} = \sum_{j=0}^{d-1} \psi_j \ket{j}$ is an eigenvector of the Pauli operator $P$ corresponding to the eigenvalue $\lambda$, i.e., $P \ket{\psi} = \lambda \ket{\psi}$. For any $j \in \mathbb{F}_d$, the action of $X^a Z^b$ on the computational basis state $\ket{j}$ yields $X^a Z^b \ket{j} = \omega^{b j} \ket{j + a}$, from which we obtain the component‑wise relation $\omega^{b j} \psi_j = \lambda \psi_{j + a}$, or equivalently $\psi_{j + a} = \omega^{b j} \lambda^{-1} \psi_j$. Iterating this recurrence $d$ times leads to the identity
\begin{equation}
\psi_j=\psi_{j+da}=\left(  \prod_{k=0}^{d-1}\omega^{b(j+ka)}\lambda^{-1}  \right)\psi_j,
\end{equation}

Thus, we have $1 = \omega^{b\sum_{k=0}^{d-1}(j+ka)} \lambda^{-d}$. Note that $\omega^{bdj} = 1$ and $\omega^{\frac{d(d-1)}{2}} = \mathrm{e}^{(d-1)\pi \mathrm{i}} = (-1)^{d-1} = 1$. It follows that $\lambda^d = 1$, which immediately implies that the eigenvalues of $P$ are the $d$-th roots of unity. Explicitly,
\begin{equation}
\lambda_s=\omega^{s},\quad s=0,1,\dots,d-1.
\end{equation}

Hence, $\psi_{j+a}=\omega^{bj-s}\psi_j$. Let $t=a^{-1}j$ and $\phi_t=\psi_{at}$, then $\phi_{t+1}=\psi_{a(t+1)}=\omega^{abt-s}\psi_{at}=\omega^{abt-s}\phi_t$. Solving this recursive relation yields
\begin{equation}
\phi_t=\phi_0\prod_{k=0}^{t-1}\omega^{abk-s}=\phi_0\omega^{-st}\omega^{ab\frac{t(t-1)}{2}}.
\end{equation}

Consequently, the $j$-th component of the eigenvector admits the explicit form $\psi_j=\frac{1}{\sqrt{d}}\omega^{ab\frac{t(t-1)}{2}-st}$, where $\frac{1}{\sqrt{d}}$ is the normalization constant. Substituting $t=a^{-1}j$ into the above expression, the eigenvector can be written as
\begin{equation}
\ket{\psi_s}=\frac{1}{\sqrt{d}}\sum_{j=0}^{d-1}\omega^{\left( \frac{ab}{2}(a^{-1}j-1)-s \right)a^{-1}j}\ket{j}.
\end{equation}
\end{proof}

In particular, when $a = 0$, the operator $P = Z^b$ is already diagonal, and its eigenvalues are simply $\lambda_j = \omega^{b j}$ with corresponding eigenvectors $\ket{j}$. On the other hand, when $b = 0$, the eigenvalues of $P$ remain $\lambda_s = \omega^{s}$, and the corresponding eigenvectors take the form $\ket{\psi_s}=\frac{1}{\sqrt{d}}\sum_{j=0}^{d-1}\omega^{-sa^{-1}j}\ket{j}$, which are precisely the Fourier basis vectors (up to a relabeling of the basis index).

For any $k \in \mathbb{F}_d$, the eigenvectors of the Pauli operator $X Z^{k}$ are given explicitly by
\begin{equation}
\ket{\varphi_{k, s}} := \frac{1}{\sqrt{d}}\sum_{j=0}^{d-1}{w^{\frac{k}{2}j^2-(\frac{k}{2} + s)j}\ket{j}},
\end{equation}
where $s\in\mathbb{F}_d$. Remarkably, these eigenvectors exhibit several unexpected and rather intriguing properties.

\begin{lemma}\label{cup stabilizer state}
Let $d>2$ be a prime, and $\mathcal{S}$ the collection of all stabilizer states in $\mathbb{C}^d$. Then $\mathcal{S} = \bigcup_{k = 0}^d{\mathcal{S}_k}$,
where
\begin{equation}
\mathcal{S}_k = \{\ket{\phi_{k,s}} : s \in \mathbb{F}_d\}~
\text{with}~
    \ket{\phi_{k,s}} = \begin{cases}
        \ket{\varphi_{k, s}},& \text{if } k \in \mathbb{F}_d; \\
        \ket{s},& \text{if } k = d.
    \end{cases}
\end{equation}
\end{lemma}

\begin{lemma}[Mutually unbiased bases]\label{MUB}
Let $d>2$ be a prime. For any $k \neq l$ and $s,t \in \mathbb{F}_d$, it holds that $\left|\braket{\phi_{k,s}|\phi_{l,t}}\right|^2 = \frac{1}{d}$.
\end{lemma}

\begin{lemma}\label{sum 2}
Let $d>2$ be a prime, and $\mathcal{S}$ the collection of all stabilizer states in $\mathbb{C}^d$. Then for any single qudit state $\ket{\psi} \in \mathbb{C}^d$, we have
\begin{equation}
\sum_{\ket{\phi} \in \mathcal{S}} \left|\braket{\psi | \phi}\right|^2 = d+1.
\end{equation}
\end{lemma}

\begin{proof}
According to Lemma~\ref{cup stabilizer state}, each $\mathcal{S}_k$ forms an orthonormal basis in $\mathbb{C}_d$. Since the sum of probabilities over any single orthonormal basis is $1$, summing over all $d+1$ bases directly yields $\sum_{\ket{\phi} \in \mathcal{S}} \left|\braket{\psi | \phi}\right|^2 = d+1$.
\end{proof}

Furthermore, we have the following simple yet intriguing observation.
\begin{lemma}\label{sum 4}
Let $d>2$ be a prime, and $\mathcal{S}$ the set of all stabilizer states in $\mathbb{C}^d$. Then, for any single qudit state $\ket{\psi} \in \mathbb{C}^d$, the following bound holds
\begin{equation}
\sum_{\ket{\phi} \in \mathcal{S}}{\left|\braket{\psi | \phi}\right|^4} = 2.
\end{equation}
\end{lemma}

\begin{proof}
By Lemma~\ref{cup stabilizer state}, we have 
\begin{equation}
    \sum_{\ket{\phi} \in \mathcal{S}}{\left|\braket{ \psi | \phi} \right|^4} = \sum_{k=0}^d\sum_{s=0}^{d-1}{\left| \braket{ \phi_{k,s} | \psi} \right|^4} = \sum_{k=0}^d\sum_{s=0}^{d-1}{\mathrm{Tr}(\phi_{k,s}\psi)^2}.
\end{equation}
Define $\psi^{(0)} = \psi - \frac{1}{d}I$, and note that $\psi^{(0)}$ is traceless, i.e., $\mathrm{Tr}(\psi^{(0)}) = 0$. Then,
\begin{align}
    \sum_{k=0}^d\sum_{s=0}^{d-1}{\mathrm{Tr}(\phi_{k,s}\psi)^2} &= \sum_{k=0}^d\sum_{s=0}^{d-1}{\mathrm{Tr}\left(\phi_{k,s}\left(\psi^{(0)} + \frac{1}{d}I\right)\right)^2} \nonumber \\
    &= \sum_{k=0}^d\sum_{s=0}^{d-1}{\left(\mathrm{Tr}(\phi_{k,s}\psi^{(0)})^2 + \frac{2}{d}\mathrm{Tr}(\phi_{k,s}\psi^{(0)}) + \frac{1}{d^2}\right)} \nonumber\\
    &= \sum_{k=0}^d\sum_{s=0}^{d-1}{\mathrm{Tr}(\phi_{k,s}\psi^{(0)})^2} + \frac{2}{d}\mathrm{Tr}\left(\left(\sum_{k=0}^d\sum_{s=0}^{d-1}{\phi_{k,s}}\right)\psi^{(0)}\right) + \frac{d+1}{d} \\
    &= \sum_{k=0}^d\sum_{s=0}^{d-1}{\mathrm{Tr}(\phi_{k,s}\psi^{(0)})^2} + \frac{2}{d}\mathrm{Tr}\left(\left((d+1)I\right)\psi^{(0)}\right) + \frac{d+1}{d} \nonumber\\
    &= \sum_{k=0}^d\sum_{s=0}^{d-1}{\mathrm{Tr}(\phi_{k,s}\psi^{(0)})^2} + \frac{d+1}{d}.\nonumber
\end{align}

Let $\mathsf{Herm}_d$ be the collection of $d \times d$ Hermitian matrices. Equipped with the Hilbert-Schmidt inner product $\langle A, B \rangle = \mathrm{Tr}(A^\dagger B)$, $\mathsf{Herm}_d$ forms a $d^2$-dimensional Euclidean space. Let $\mathsf{Herm}_d^0$ be the collection of $d \times d$ traceless Hermitian matrices, which is a $(d^2-1)$-dimensional subspace of $\mathsf{Herm}_d$. Let $\mathcal{H}_k = \mathrm{span}\left(\phi_{k,0},\dots, \phi_{k,d-1}\right)$ be the subspace spanned by the projectors $\phi_{k,0},\dots, \phi_{k,d-1}$ for some $k \in \{0,1,\dots,d\}$. Since these projectors form an orthogonal basis of $\mathcal{H}_k$, we have $\dim(\mathcal{H}_k) = d$. Let $\mathcal{H}_k^0 = \mathcal{H}_k \cap \mathsf{Herm}_d^0$, and note that $I = \sum_{s = 0}^{d-1}{\phi_{k,s}} \in \mathcal{H}_k$ and $\mathsf{Herm}_d^0 = \mathrm{span}(I)^{\perp}$. Hence, we obtain $\mathcal{H}_k + \mathsf{Herm}_d^0 = \mathsf{Herm}_d$, and 
\begin{small}
\begin{equation}
    \dim(\mathcal{H}_k^0) = \dim(\mathcal{H}_k) + \dim(\mathsf{Herm}_d^0) - \dim(\mathcal{H}_k + \mathsf{Herm}_d^0) = d-1.
\end{equation}
\end{small}

Define $\psi_{k}^{(0)} = \sum_{s=0}^{d-1}{\mathrm{Tr}(\phi_{k,s}\psi^{(0)})\phi_{k,s}}$ as the orthogonal projection of $\psi^{(0)}$ onto the subspace $\mathcal{H}_k$. It's easy to verify that $\psi_k^{(0)}$ is traceless:
\begin{equation}
    \mathrm{Tr}(\psi_k^{(0)}) = \sum_{s=0}^{d-1}\mathrm{Tr}(\phi_{k,s}\psi^{(0)})\mathrm{Tr}(\phi_{k,s}) = \mathrm{Tr}\left(\psi^{(0)}\sum_{s=0}^{d-1}\phi_{k,s}\right) = \mathrm{Tr}(\psi^{(0)}) = 0.
\end{equation}
Thus, $\psi_k^{(0)} \in \mathcal{H}_k^0$ for all $k \in \{0,1,\dots,d\}$. 

Next, we will prove that for any $k \neq l$, the subspaces $\mathcal{H}_k^0$ and $\mathcal{H}_l^0$ are orthogonal. For any $A \in \mathcal{H}_k^0$ and $B \in \mathcal{H}_l^0$, there exists  $a_0,\dots, a_{d-1}, b_0,\dots, b_{d-1} \in \mathbb{R}$ such that $A = \sum_{s=0}^{d-1}{a_s\phi_{k,s}}$ and $B = \sum_{l=0}^{d-1}{b_s\phi_{l,t}}$. So, the inner product between $A$ and $B$ is
\begin{equation}
    \langle A, B \rangle = \sum_{s=0}^{d-1}\sum_{t=0}^{d-1} a_s b_t \mathrm{Tr}(\phi_{k,s}\phi_{l,t}) = \frac{1}{d}\sum_{s=0}^{d-1}\sum_{t=0}^{d-1} a_s b_t  = \frac{1}{d}\left(\sum_{s=0}^{d-1} a_s\right)\left(\sum_{t=0}^{d-1} b_t\right) = 0.
\end{equation}
Here we use $\mathrm{Tr}(\phi_{k,s}\phi_{l,t})= \frac{1}{d}$ from Lemma~\ref{MUB}. Therefore, $\mathcal{H}_k^0$ and $\mathcal{H}_l^0$ are orthogonal for any $k \neq l$. This implies that 
\begin{equation}
    \dim\left(\sum_{k=0}^d\mathcal{H}_k^0 \right) = \sum_{k=0}^d{\dim(\mathcal{H}_k^0)} = d^2 - 1 = \dim(\mathsf{Herm}_d^0),
\end{equation}
which is equivalent to $\mathsf{Herm}_d^0 = \mathcal{H}_0^0 \oplus \cdots \oplus \mathcal{H}_d^0$.

By the (generalized) Parseval's theorem, we have $\sum_{s=0}^{d-1}{\mathrm{Tr}(\phi_{k,s}\psi^{(0)})^2} = \|\psi_{k}^{(0)}\|_F^2$, where $\|\cdot\|_F$ denotes the Frobenius norm. Then, it holds that
\begin{small}
\begin{equation}
    \sum_{k=0}^{d}\sum_{s=0}^{d-1}{\mathrm{Tr}(\phi_{k,s}\psi^{(0)})^2} = \|\psi^{(0)}\|_F^2=\mathrm{Tr}\left(\left(\psi - \frac{I}{d}\right)^2\right)=\frac{d-1}{d}.
\end{equation}
\end{small}

Therefore,
\begin{equation}
    \sum_{\ket{\phi} \in \mathcal{S}}{|\langle \psi | \phi \rangle|^4} = \sum_{k=0}^d\sum_{s=0}^{d-1}{\mathrm{Tr}(\phi_{k,s}\psi^{(0)})^2} + \frac{d+1}{d} = 2.
\end{equation}
\end{proof}

As an immediate corollary, the theorem establishes that the set $\mathcal{S}$ of stabilizer states forms a $2$-design. We note that a closely related result on $t$-design was previously obtained in Ref.~\cite{KR05}. In contrast to the approach taken there, our proof relies solely on elementary linear algebra, rendering the argument considerably more direct and accessible.

With the preceding results in hand, we are now in a position to establish a lower bound on the stabilizer fidelity of an arbitrary single qudit state.

\begin{theorem}\label{bound F 1qudit}
Let $d>2$ be a prime, and $\mathcal{S}$ the set of all stabilizer states in $\mathbb{C}^d$. Then for any single qudit state $\ket{\psi} \in \mathbb{C}^d$, the stabilizer fidelity $F_{\mathcal{S}}(\ket{\psi})$ is at least $\frac{2}{d+1}$.
\end{theorem}

\begin{proof}
According to the definition of stabilizer fidelity, we have 
\begin{equation}
    \sum_{\ket{\phi} \in \mathcal{S}}{\left|\braket{ \psi | \phi} \right|^4} \le F_{\mathcal{S}}(\ket{\psi})\sum_{\ket{\phi} \in \mathcal{S}}{\left|\braket{ \psi | \phi} \right|^2}.
\end{equation}
Combining Lemma~\ref{sum 2} with Lemma~\ref{sum 4}, we obtain the following lower bound on $F_{\mathcal{S}}(\ket{\psi})$:
\begin{equation}
F_{\mathcal{S}}(\ket{\psi}) \geq \frac{\sum_{\ket{\phi} \in \mathcal{S}}{\left|\braket{ \psi | \phi} \right|^4}}{\sum_{\ket{\phi} \in \mathcal{S}}{\left|\braket{ \psi | \phi} \right|^2}} = \frac{2}{d+1}.
\end{equation}
\end{proof}

\subsection{Proof of the anti-concentration property}\label{subsection property 2}


We now prove Theorem~\ref{theorem B_psi}. For convenience, write $F \coloneqq F_{\mathcal{S}}(\ket{\psi})$ and $S \coloneqq \mathrm{Weyl}(\ket{\phi})$, where $\ket{\phi}$ is a stabilizer state achieving fidelity $F$ with $\ket{\psi}$.

The overall proof strategy adapts the qubit approach of~\cite{GIKL24b} to the marginal distribution $\mathbf{B}_\psi$ of skewed Bell difference sampling. Namely, the bound is first established for $\ket{\phi} = \ket{0^n}$ and then extended to arbitrary stabilizer states by Clifford conjugation. The key departure from the qubit setting lies in lower-bounding the sum of characteristic-function values over $S \setminus T$ for any proper subspace $T \subsetneq S$, where we invoke the single-qudit stabilizer fidelity bound from Theorem~\ref{bound F 1qudit} to circumvent a multivariate optimization that does not admit a closed-form solution for general $d>2$.

\begin{lemma}\label{lemma c_psi}
Let $\ket{\psi} \in (\mathbb{C}^d)^{\otimes n}$ be a pure state with stabilizer fidelity $F$, where $d>2$ is a prime. Define $S = 0^n \times \mathbb{F}_d^{n} = \mathrm{Weyl}(\ket{0^n})$, and let $T = 0^{n+1} \times \mathbb{F}_d^{n-1}$ be a proper subspace of $S$. Then the characteristic function $c_{\psi}(\mathbf{x})$ satisfies 
\begin{equation}
\sum_{\mathbf{x} \in S \setminus T} c_{\psi}(\mathbf{x}) \geq \frac{d^{\frac{n}{2}-1}(d-1)}{2} \cdot F.
\end{equation}
\end{lemma}

\begin{proof}
Recall that the Pauli $Z$ gate admits the diagonal representation $Z = \operatorname{diag}(1, \omega, \dots, \omega^{d-1})$. Consequently, the sum of its inverse powers evaluates to $\sum_{k=0}^{d-1}(Z^\dagger)^k=\mathrm{diag}(d,0,\dots,0)=d\ket{0}\bra{0}$.

Let $\mathbb{F}_{d \backslash 0}=\{1,2,\dots,d-1\}$. Then, we have 
\begin{align}\label{eq c_psi}
\sum_{\mathbf{x} \in S \backslash T}c_{\psi}(\mathbf{x})
&= d^{-\frac{n}{2}} \sum_{\mathbf{v}\in \mathbb{F}_{d \backslash 0} \times \mathbb{F}_d^{n-1}} \bra{\psi}W_{0^n,\mathbf{v}}^\dagger\ket{\psi} \nonumber \\
&= d^{-\frac{n}{2}} \sum_{\small{\begin{smallmatrix}
v_1\in \mathbb{F}_{d \backslash 0},\\
v_2,\dots,v_n \in \mathbb{F}_d
\end{smallmatrix}}} \bra{\psi} (Z^\dagger)^{v_1} \otimes (Z^\dagger)^{v_2} \otimes \cdots \otimes (Z^\dagger)^{v_n} \ket{\psi} \nonumber \\
&= d^{-\frac{n}{2}} \bra{\psi} (d\ket{0}\bra{0}-I) \otimes (d\ket{0}\bra{0} \otimes \cdots \otimes d\ket{0}\bra{0}) \ket{\psi}  \\
&= d^{\frac{n}{2}-1} \bra{\psi} (d\ket{0}\bra{0}-I)  \otimes (\ket{0}\bra{0})^{\otimes (n-1)} \ket{\psi} \nonumber \\
&= d^{\frac{n}{2}-1} \left( (d-1)|\alpha_0|^2- \sum_{k=1}^{d-1}|\alpha_k|^2 \right), \nonumber
\end{align}
where $\alpha_k=\langle \psi | k0^{n-1} \rangle$ for $k \in \mathbb{F}_d$.

Since the stabilizer fidelity is given by $F = \left|\braket{\psi|0^n}\right|^2 = |\alpha_0|^2$, by the maximality of $F$, we must have $|\alpha_0|^2 \geq \left|\bra{\psi}(\ket{\varphi 0^{n-1}})\right|^2$ for every single-qudit stabilizer state $\ket{\varphi} \in \mathbb{C}^d$, where we write $\ket{\varphi \, 0^{n-1}} := \ket{\varphi} \otimes \ket{0^{n-1}}$ for brevity. Now consider the state $\ket{\psi_1} = \frac{1}{\sqrt{A}}\sum_{k=0}^{d-1}\alpha_k \ket{k}$ with $A =\sum_{k=0}^{d-1}|\alpha_k|^2$. From the inequality above, we obtain
\begin{equation}
|\langle \psi_1 | 0 \rangle|^2 = \frac{|\alpha_0|^2}{A} \ge \frac{|\langle \psi | \varphi0^{n-1} \rangle|^2}{A} = |\langle \psi_1 | \varphi \rangle|^2.
\end{equation}
This implies that the stabilizer fidelity of $\ket{\psi_1}$ is $F_{\mathcal{S}}(\ket{\psi_1}) = |\langle \psi_1 | 0 \rangle|^2$. By Theorem~\ref{bound F 1qudit}, we can get $F_{\mathcal{S}}(\ket{\psi_1}) \geq \frac{2}{d+1}$.

Notice that $A = |\alpha_0|^2 + \sum_{k=1}^{d-1}|\alpha_k|^2$ and $|\alpha_0|^2 = A \cdot F_{\mathcal{S}}(\ket{\psi_1})$. Then,
\begin{equation}
\frac{\sum_{k=1}^{d-1}{|\alpha_k|^2}}{(d-1)|\alpha_0|^2} = \frac{A(1-F_{\mathcal{S}}(\ket{\psi_1}))}{A(d-1)F_{\mathcal{S}}(\ket{\psi_1})} \leq \frac{1-\frac{2}{d+1}}{(d-1)\frac{2}{d+1}} = \frac{1}{2}.
\end{equation}

It therefore follows that
\begin{align}
\sum_{\mathbf{x} \in S \setminus T}{c_{\psi}(\mathbf{x})} &= d^{\frac{n}{2}-1}\left((d-1)|\alpha_0|^2 - \sum_{k=1}^{d-1}{|\alpha_k|^2}\right) \nonumber \\
& \ge d^{\frac{n}{2}-1}\left((d-1)|\alpha_0|^2 - \frac{1}{2}(d-1)|\alpha_0|^2\right) \\
&= \frac{d^{\frac{n}{2}-1}(d-1)}{2} \cdot F. \nonumber
\end{align}
\end{proof}

Since $p_{\psi}(\mathbf{x}) = |c_{\psi}(\mathbf{x})|^2$, we may apply Lemma~\ref{lemma c_psi} to bound $p_{\psi}(\mathbf{x})$ and thereby characterize the behavior of the marginal distribution $\mathbf{B}_{\psi}(\mathbf{x})$ over the set $S \setminus T$.

\begin{theorem}\label{theorem B_psi 0}
Let $\ket{\psi}\in (\mathbb{C}^d)^{\otimes n}$ be an unknown quantum state with stabilizer fidelity $F$, where $d>2$ is a prime. Let $\ket{0^n}=\arg \max_{\ket{\varphi}\in \mathcal{S}_d^n} \left|\braket{\varphi|\psi}\right|^2$.  Define $S=\mathrm{Weyl}(\ket{0^n})$, and let $T=0^{n+1} \times \mathbb{F}_d^{n-1}$ be a proper subspace of $S$. Then, we have
\begin{equation}
\sum_{\mathbf{x} \in S \backslash T} \mathbf{B}_{\psi}(\mathbf{x}) \geq \frac{(d-1)F^6}{4d^3}.
\end{equation}
\end{theorem}

\begin{proof}
Note that $S\setminus T$ is not a proper subspace of $\mathbb{F}_d^{2n}$, but for any $\mathbf{Y}_1 \in T$, we have $\mathbf{Y}_1 + (S\setminus T) = S\setminus T$. Consequently,
\begin{align}
\sum_{\mathbf{x} \in S \backslash T} \mathbf{B}_{\psi}(\mathbf{x}) &= \sum_{\mathbf{x} \in S \backslash T} \sum_{\mathbf{Y} \in \mathbb{F}_d^{2n \times 2}} p_\psi(\mathbf{x} + \mathbf{Y}_1) \prod_{i=1}^2 p_\psi((\mathbf{Y}\mathbf{R})_i) \nonumber\\
&= \sum_{\mathbf{Y} \in \mathbb{F}_d^{2n \times 2}} \prod_{i=1}^2 p_\psi((\mathbf{Y}\mathbf{R})_i)  \sum_{\mathbf{x} \in S \backslash T} p_\psi(\mathbf{x} + \mathbf{Y}_1) \\
&\geq \sum_{\mathbf{Y} \in T^{\otimes 2}} \prod_{i=1}^2 p_\psi((\mathbf{Y}\mathbf{R})_i) \cdot \sum_{\mathbf{x} \in S \backslash T} p_\psi(\mathbf{x}). \nonumber
\end{align}

Since $T$ is a proper subspace of $S \subseteq \mathbb{F}_d^{2n}$, by Lemma~\ref{subspace lowerbound} and Lemma~\ref{T'=T}, we have
\begin{align}
\sum_{\mathbf{x} \in S \backslash T} \mathbf{B}_{\psi}(\mathbf{x}) &\geq \prod_{i=1}^2 \left( \sum_{\mathbf{Y}_i \in T} p_{\psi}(\mathbf{Y}_i) \right) \cdot \sum_{\mathbf{x} \in S \backslash T} p_\psi(\mathbf{x}) \nonumber \\
&\geq \left( \frac{|T|}{d^n}F^2 \right)^2 \cdot \sum_{\mathbf{x} \in S \backslash T} |c_\psi(\mathbf{x})|^2 \nonumber \\
&\geq \frac{F^4}{(d-1)d^{n-1} \cdot d^2} \cdot \left( \sum_{\mathbf{x} \in S \backslash T} |c_\psi(\mathbf{x})| \right)^2\\
&\geq \frac{F^4}{(d-1)d^{n+1}} \cdot \left( \sum_{\mathbf{x} \in S \backslash T} c_\psi(\mathbf{x}) \right)^2, \nonumber
\end{align}
where in Line 3 we utilize the Cauchy-Schwarz inequality. 

Finally, by applying Lemma~\ref{lemma c_psi}, we obtain the desired conclusion.
\begin{equation}
\sum_{\mathbf{x} \in S \backslash T} \mathbf{B}_{\psi}(\mathbf{x}) \geq \frac{F^4}{(d-1)d^{n+1}} \cdot \frac{(d-1)^2 d^{n-2} F^2}{4} = \frac{(d-1)F^6}{4d^3}.
\end{equation}
\end{proof}

The generalization of the probabilistic properties from $\mathrm{Weyl}(\ket{0^n})$ to an arbitrary $\mathrm{Weyl}(\ket{\phi})$ hinges on identifying a Clifford circuit that maps the former subspace to the latter. Moreover, as we shall verify, the marginal probability distribution of interest remains invariant under the action of any such circuit.

\begin{lemma}\label{lemma circuit}
Let $\ket{\phi}\in (\mathbb{C}^d)^{\otimes n}$ be an $n$-qudit stabilizer state, where $d>2$ is a prime. Let $S = \mathrm{Weyl}(\ket{\phi})$ and $T \subseteq S$ be a subspace of size $d^{n-t}$ with $t>0$. Then there exists a Clifford circuit $C$ such that
\begin{equation}
C\ket{\phi}=\ket{0^n},~ C(S)=0^n\times \mathbb{F}_d^n,~ C(T)=0^{n+t}\times \mathbb{F}_d^{n-t}.
\end{equation}
\end{lemma}

\begin{proof}
Since the Clifford group acts transitively on stabilizer states, there exists a Clifford circuit $C_1$ such that $C_1 \ket{\phi} = \ket{0^n}$. Then $C_1(S)=\bigl\{\mathbf{x}\in \mathbb{F}_d^{2n}: \bigl|\bra{\phi}C_1^\dagger W_{\mathbf{x}} C_1\ket{\phi}\bigr|=1\bigr\}=0^n \times \mathbb{F}_d^n$, and $C_1(T)\subseteq C_1(S)$.

Clearly, $C_1(T)$ is an $(n-t)$-dimensional subspace. By Theorem~\ref{theorem circuit}, there exists a Clifford circuit $C_2$ such that $C_2(C_1(T))=0^{n+t} \times \mathbb{F}_d^{n-t}$, and this circuit does not affect $0^n \times \mathbb{F}_d^n$.

Therefore, taking $C=C_2C_1$ we have $C\ket{\phi}=\ket{0^n}$, $C(S)=0^n\times \mathbb{F}_d^n$, and $C(T)=0^{n+t}\times \mathbb{F}_d^{n-t}$.
\end{proof}

\begin{lemma}\label{B_psi unchanged}
Let $\ket{\psi} \in (\mathbb{C}^d)^{\otimes n}$ be an arbitrary pure state, and $C$ any Clifford circuit. If we set $\ket{\phi} = C \ket{\psi}$, then the marginal probability distributions $\mathbf{B}_{\phi}$ and $\mathbf{B}_{\psi}$ are related by $\mathbf{B}_{\phi}(C(\mathbf{x}))=\mathbf{B}_{\psi}(\mathbf{x})$ for all $\mathbf{x}\in\mathbb{F}_d^{2n}$.
\end{lemma}

\begin{proof}
A fundamental property of the Clifford action on the space $\mathbb{F}_d^{2n}$ is its linearity. Indeed, for any two Weyl operators $W_{\mathbf{x}}$ and $W_{\mathbf{y}}$ with $\mathbf{x}, \mathbf{y} \in \mathbb{F}_d^{2n}$, the relation $W_{\mathbf{x}} W_{\mathbf{y}} = \kappa^{[\mathbf{x}, \mathbf{y}]} W_{\mathbf{x} + \mathbf{y}}$ implies that $C(\mathbf{x})+C(\mathbf{y}) = C(\mathbf{x+y})$. 

On the other hand, fix $a \in \mathbb{F}_d$ and suppose $C(a \mathbf{x}) = \mathbf{y}$, which by definition means that there exists a $s \in \mathbb{F}_d$ such that $C W_{a\mathbf{x}} C^\dagger = \omega^s W_{\mathbf{y}}$. Since $W_{ax}=W_{\mathbf{x}}^a$ for any $a\in\mathbb{F}_d$, we obtain $C W_{\mathbf{x}} C^\dagger = \omega^{a^{-1}s} W_{a^{-1}\mathbf{y}}$. Consequently, $a^{-1}\mathbf{y} = C(\mathbf{x})$, i.e., $C(a\mathbf{x}) = aC(\mathbf{x})$.

Furthermore, for any $\mathbf{x}\in \mathbb{F}_d^{2n}$ we have $p_{\phi}(C(\mathbf{x}))=p_{\psi}(\mathbf{x})$, which is because 
\begin{equation}
d^np_{\phi}(C(\mathbf{x}))=\left|\bra{\phi}CW_{\mathbf{x}}C^\dagger\ket{\phi}\right|^2=\left|\bra{\psi}W_{\mathbf{x}}\ket{\psi}\right|^2=d^np_{\psi}(\mathbf{x}).
\end{equation}
Therefore, it holds that
\begin{align}
\mathbf{B}_{\phi}(C(\mathbf{x})) &= \sum_{\mathbf{Y} \in \mathbb{F}_d^{2n \times 2}} p_\phi(C(\mathbf{x}) + \mathbf{Y}_1) \prod_{i=1}^2 p_\phi((\mathbf{Y}\mathbf{R})_i) \nonumber \\
&= \sum_{\mathbf{Y} \in \mathbb{F}_d^{2n \times 2}} p_\phi(C(\mathbf{x} + \mathbf{Y}_1)) \prod_{i=1}^2 p_\phi(C((\mathbf{Y}\mathbf{R})_i)) \\
&= \sum_{\mathbf{Y} \in \mathbb{F}_d^{2n \times 2}} p_\psi(\mathbf{x} + \mathbf{Y}_1) \prod_{i=1}^2 p_\psi((\mathbf{Y}\mathbf{R})_i) \nonumber \\
&= \mathbf{B}_{\psi}(\mathbf{x}). \nonumber
\end{align}
\end{proof}

From the preceding results, we can now establish Theorem~\ref{theorem B_psi}, the anti-concentration property of the marginal distribution $\mathbf{B}_{\psi}$.


\begin{proof}[Proof of Theorem~\ref{theorem B_psi}]
Note that $|T| < d^n$. Consequently, by Lemma~\ref{lemma circuit}, there exists a Clifford circuit $C$ such that $C \ket{\phi} = \ket{0^n}$, $C(S) = 0^n \times \mathbb{F}_d^n$, and $C(T) \subseteq 0^{n+1} \times \mathbb{F}_d^{n-1}$. Now set $\ket{\psi'} = C \ket{\psi}$. Applying Lemma~\ref{B_psi unchanged}, we obtain
\begin{equation}
\sum_{\mathbf{x} \in S \backslash T} \mathbf{B}_{\psi}(\mathbf{x})=\sum_{\mathbf{x} \in C(S \backslash T)} \mathbf{B}_{\psi'}(\mathbf{x})\geq \sum_{\mathbf{x} \in 0^n\times \mathbb{F}_{d\backslash 0} \times \mathbb{F}_d^{n-1}} \mathbf{B}_{\psi'}(\mathbf{x}).
\end{equation}
According to Theorem~\ref{theorem B_psi 0}, we have
\begin{equation}
\sum_{\mathbf{x} \in S \backslash T} \mathbf{B}_{\psi}(\mathbf{x})\geq \frac{(d-1)F_{\mathcal{S}}(\ket{\psi'})^6}{4d^3}=\frac{(d-1)F_{\mathcal{S}}(\ket{\psi})^6}{4d^3},
\end{equation}
where $F_{\mathcal{S}}(\ket{\psi'})=F_{\mathcal{S}}(\ket{\psi})$.
\end{proof}

\section{Agnostic learning via stabilizer bootstrapping} \label{bootstrapping}

The stabilizer bootstrapping algorithm, recently introduced by Chen et~al.~\cite{CGYZ25}, provides a unified framework for agnostic learning of quantum states on qubit systems. Specifically, the algorithm iteratively amplifies fidelity with a target state by first accumulating a commuting set of high-correlation Pauli projectors via Bell difference sampling. If this set is incomplete, the distribution's anti-concentration ensures that sampling yields a low-correlation projector which nonetheless stabilizes the target; measuring and post-selecting on this projector then boosts fidelity by a constant factor. The process recurses on the amplified state, requiring only a logarithmically small number of rounds to converge to the desired stabilizer state with high probability. 

In the present work, we extend this methodology to the agnostic learning of $n$-qudit stabilizer states. To lay the groundwork, this section collects two elementary lemmas that will be essential for the analysis that follows.


\begin{lemma}[\cite{GNW21}]\label{lemma uncertainty relation}
For $\mathbf{x},\mathbf{y} \in \mathbb{F}_d^{2n}$, let $\ket{\psi} \in (\mathbb{C}^d)^{\otimes n}$ be a $n$-qudit pure state such that $\left|\bra{\psi}W_{\mathbf{x}}\ket{\psi}\right|^2\geq 1-\frac{1}{4d^2}$ and $\left|\bra{\psi}W_{\mathbf{y}}\ket{\psi}\right|^2\geq 1-\frac{1}{4d^2}$. Then $W_{\mathbf{x}}$ and $W_{\mathbf{y}}$ must commute.
\end{lemma}

\begin{lemma}[\cite{ADIS25}]\label{lemma bound D}
Let $d>2$ be a prime, $\varepsilon, \delta \in (0,1)$, and $\mathcal{D}$ a distribution over $\mathbb{F}_d^n$. Denote by $A \subseteq \mathbb{F}_d^n$ the subspace spanned by $m$ independent samples from $\mathcal{D}$. If $m \geq \frac{2n + 2 \log(1/\delta)}{\varepsilon}$, then with probability at least $1 - \delta$ we have $\sum_{x \in A} \mathcal{D}(x) \geq 1 - \varepsilon$.
\end{lemma}

In the following, we provide a step‑by‑step analysis of the stabilizer bootstrapping algorithm for identifying the stabilizer state $\ket{\phi}$ that best approximates an unknown pure state $\ket{\psi}$.

{\bf Step 1: Find a high-correlation family}

If $A \subseteq \mathbb{F}_d^{2n}$ is an isotropic subspace of size $d^n$, then we call $A$ a stabilizer family.

\begin{definition}[$\varepsilon$-high-correlation family]
Suppose $\ket{\psi}$ is an unknown $n$-qudit quantum state. If $\left|\bra{\psi}W_{\mathbf{y}}\ket{\psi}\right|^2 > 1-\frac{1}{12d^2}$, then $\mathbf{y}\in \mathbb{F}_d^{2n}$ is called a high-correlation Pauli string. Conversely, $\mathbf{y}\in \mathbb{F}_d^{2n}$ is called a low-correlation Pauli string. We say that a set of Pauli strings $T \subseteq \mathbb{F}_d^{2n}$ is an $\varepsilon$-high-correlation family if
\begin{equation}
\Pr_{\mathbf{y}\sim \mathbf{B}_{\psi}} \left[ \left|\bra{\psi}W_{\mathbf{y}}\ket{\psi}\right|^2 > 1-\tfrac{1}{12d^2} \land \mathbf{y}\notin T \right] \leq \varepsilon.
\end{equation}
A basis of an $\varepsilon$-high-correlation family is called an $\varepsilon$-high-correlation basis.
\end{definition}

To produce a high‑correlation family, we first draw a collection of Pauli strings via $\mathbf{B}_{\psi}$ sampling. For each sampled string, we estimate its correlation with the unknown state using the SWAP test. Strings whose estimation exceeds the prescribed threshold $1 - \frac{1}{6d^2}$ are retained, and we compute a basis $H$ for the subspace they span. If $|H| < n$, we complete $H$ to a full basis of a stabilizer family by adding commuting Pauli strings to form a maximally isotropic set. The resulting basis is then returned as the output. A complete description of the procedure is given in Algorithm~\ref{algorithm learn H}.

\begin{algorithm}[!t]
\caption{Learn high-correlation Pauli strings}
\label{algorithm learn H}
\KwIn{$\varepsilon, \delta  > 0$, copies of an $n$-qudit state $\ket{\psi} \in (\mathbb{C}^d)^{\otimes n}$ where $d>2$ is a prime}
\KwOut{An $\varepsilon$-high-correlation basis $H$ of a stabilizer family with probability at least $1-\delta$}

Perform $\mathbf{B}_{\psi}$ sampling for $m = \frac{8(4n + \log(\frac{3}{\delta}))}{\varepsilon}$ times to obtain a set of Pauli strings $H_0 = \{\mathbf{h}_1,\dots, \mathbf{h}_m\}$.\label{algorithm1 Line1}

Using the SWAP test (Theorem~\ref{theorem M estimations}) to estimate $\left|\bra{\psi}W_{\mathbf{h}_i}\ket{\psi}\right|^2$, and obtain $\hat{E}_i$ such that with probability at least $1-\frac{\delta}{3}$,  $\bigl|\hat{E}_i-\left|\bra{\psi}W_{\mathbf{h}_i}\ket{\psi}\right|^2\bigr| \leq \frac{1}{12d^2}$ for all $i$.\label{algorithm1 Line2}

$H_1=\{ \mathbf{h}_i\in H_0 : \hat{E}_i > 1-\frac{1}{6d^2}\}$. \label{algorithm1 Line3}

Let $H$ be a basis for $\mathrm{span}(H_1)$. Abort if $H$ contains anti-commuting Pauli strings.\label{algorithm1 Line4}

If $|H| < n$, add some commuting Pauli string to $H$ to make it the basis of a stabilizer family.\label{algorithm1 Line5}

\KwRet{$H$}
\end{algorithm}

\begin{theorem}\label{theorem learn H}
With probability at least $1-\delta$, the output $H$ of Algorithm~\ref{algorithm learn H} forms an $\varepsilon$-high-correlation basis. The sample complexity is $O\left( \frac{d^4}{\varepsilon}\left(n+\log \frac{1}{\delta}\right) \left( \log n + \log \frac{1}{\delta} +\log \frac{1}{\varepsilon} \right) \right)$, and the time complexity is $O\left( \frac{n}{\varepsilon}\left(n+\log \frac{1}{\delta}\right)\left(n+d^4 \left(\log \frac{1}{\delta} + \log \frac{1}{\varepsilon} \right) \right) \right)$.
\end{theorem}

\begin{proof}
Let $G=\{ \mathbf{g}\in \mathbb{F}_d^{2n} : \left|\bra{\psi}W_{\mathbf{g}}\ket{\psi}\right|^2 > 1-\frac{1}{12d^2} \}$, $H_h = H_0 \cap G$ and $p = \Pr_{\mathbf{g} \sim \mathbf{B}_{\psi}}\left[\mathbf{g} \in G\right]$. Let $\mathcal{D}_h$ denote the probability distribution of $\mathbf{B}_{\psi}$ conditioned on the set $G$ of high-correlation Pauli operators, i.e.,
\begin{equation}
\mathcal{D}_h(\mathbf{g})= 
\begin{cases}
    \dfrac{\mathbf{B}_{\psi}(\mathbf{g})}{p}, & \mathbf{g}\in G; \\[1em]
    0, & \mathbf{g} \notin G.
\end{cases}
\end{equation}
We will show that with probability at least $1-\frac{2\delta}{3}$, $\Pr_{\mathbf{g} \sim \mathbf{B}_{\psi}}\left[\mathbf{g}\in G \land \mathbf{g} \notin \mathrm{span}(H_h)\right] \leq \varepsilon$.

If $p \leq \varepsilon$, then the above result naturally holds. If $p > \varepsilon$ then according to the Chernoff bound,
\begin{equation}
\Pr_{H_0 \sim \mathbf{B}_{\psi}^{\otimes m}}\left[ |H_h| \leq \frac{pm}{2}\right] \leq e^{-\frac{pm}{8}} \leq e^{-\frac{\varepsilon m}{8}} \leq \frac{\delta}{3}.
\end{equation}
Each element in $H_h$ can be regarded as an independent sample from $\mathcal{D}_h$. According to Lemma~\ref{lemma bound D}, if $|H_h| \geq \frac{pm}{2} \geq \frac{4n + 2\log(\frac{3}{\delta})}{\varepsilon / p}$, then with probability at least $1 - \frac{\delta}{3}$ we have
\begin{equation}
\Pr_{\mathbf{g} \sim \mathcal{D}_h}\left[\mathbf{g} \in \mathrm{span}(H_h)\right] \geq 1 - \varepsilon / p,
\end{equation}
Thus, 
\begin{equation}
\Pr_{\mathbf{g}\sim \mathbf{B}_{\psi}}\left[\mathbf{g}\in G \land \mathbf{g} \notin \mathrm{span}(H_h)\right] = \Pr_{\mathbf{g}\sim \mathbf{B}_{\psi}}\left[\mathbf{g}\in G\right]\Pr_{\mathbf{g}\sim \mathbf{B}_{\psi}}\left[\mathbf{g} \notin \mathrm{span}(H_h)| \mathbf{g}\in G\right] = p\Pr_{\mathbf{g}\sim \mathcal{D}_h}\left[\mathbf{g} \notin \mathrm{span}(H_h)\right] \leq \varepsilon,
\end{equation}
with probability at least $1-\frac{2\delta}{3}$.

So, with probability at least $1-\delta$, it holds that $\Pr_{\mathbf{g}\sim \mathbf{B}_{\psi}}\left[\mathbf{g}\in G \land \mathbf{g} \notin \mathrm{span}(H_h)\right] \leq \varepsilon$ and $\bigl|\hat{E}_i-\left|\bra{\psi}W_{\mathbf{h}_i}\ket{\psi}\right|^2\bigr|\leq \frac{1}{12d^2}$ for all $i$. Next, we will show that under this condition, the set $H$ must be an $\varepsilon$-high-correlation basis. For any $\mathbf{h}_i \in H_h$, we have $\hat{E}_i \geq \left|\bra{\psi}W_{\mathbf{h}_i}\ket{\psi}\right|^2 -\frac{1}{12d^2} > 1-\frac{1}{6d^2}$, so $H_h \subset H_1$. Moreover, for any $\mathbf{h}_i\in H_1$, we have $\left|\bra{\psi}W_{\mathbf{h}_i}\ket{\psi}\right|^2 \geq \hat{E}_i -\frac{1}{12d^2} \geq 1-\frac{1}{4d^2}$. Hence, by Lemma~\ref{lemma uncertainty relation}, all elements in $H_1$ and $H$ are mutually commuting. Consequently, the algorithm does not abort and instead outputs a basis $H$ of a stabilizer family, and $\mathrm{span}(H_h) \subset \mathrm{span}(H)$. Finally,
\begin{equation}
\Pr_{\mathbf{g} \sim \mathbf{B}_{\psi}}\left[\left|\bra{\psi}W_{\mathbf{g}}\ket{\psi}\right|^2>1-\tfrac{1}{12d^2} \land \mathbf{g} \notin \mathrm{span}(H)\right] \leq \Pr_{\mathbf{g} \sim \mathbf{B}_{\psi}}\left[\left|\bra{\psi}W_{\mathbf{g}}\ket{\psi}\right|^2>1-\tfrac{1}{12d^2} \land \mathbf{g} \notin \mathrm{span}(H_h)\right] \leq \varepsilon.
\end{equation}

The sample complexity arises from the $\mathbf{B}_{\psi}$ sampling in Line~\ref{algorithm1 Line1} and the correlation estimation in Line~\ref{algorithm1 Line2}. So the total number of samples is 
\begin{equation}
8m+O\left( 144d^4m\log \frac{3m}{\delta} \right)=O\left( \frac{d^4}{\varepsilon}\left(n+\log \frac{1}{\delta}\right) \left( \log n + \log \frac{1}{\delta} +\log \frac{1}{\varepsilon} \right) \right).
\end{equation}

The $\mathbf{B}_{\psi}$ sampling requires $O(mn)$ time. According to Theorem~\ref{theorem estimation}, the $m$ correlation estimations require $O(d^4 m n \log \frac{m}{\delta})$ time. Finding $H_1$ takes $O(m)$ time. Theorem~\ref{theorem circuit} allows us to execute Line~\ref{algorithm1 Line4} and Line~\ref{algorithm1 Line5} to output a Clifford circuit $C$ such that $C^{\dagger}(0^n \times \mathbb{F}_d^n) = H$, and this process requires $O(mn^2)$ time. Therefore, the total time complexity is 
\begin{equation}
    O\left( \frac{n}{\varepsilon}\left(n+\log \frac{1}{\delta}\right)\left(n+d^4 \left(\log \frac{1}{\delta} + \log \frac{1}{\varepsilon} \right) \right) \right).
\end{equation}
\end{proof}

{\bf Step 2: If the family is complete, then directly obtain the answer.}

If $\mathrm{span}(H) = \mathrm{Weyl}(\ket{\phi})$, then measuring $\ket{\psi}$ in the joint eigenbasis of $\mathrm{span}(H)$ yields the outcome $\ket{\phi}$ with probability $\left|\braket{\psi|\phi}\right|^2 \geq \tau$.

{\bf Step 3: If the family is incomplete, sample a low-correlation operator.}

If the set of Pauli operators returned by Algorithm~\ref{algorithm learn H} is not yet complete, we must ensure that, with high probability, a Pauli string with low correlation is sampled from $\mathrm{Weyl}(\ket{\phi}) \setminus \operatorname{span}(H)$. This step increases the fidelity between the unknown state and $\ket{\phi}$, thereby facilitating the subsequent search for candidate stabilizer states.

\begin{theorem}\label{theorem lower bound low correlation}
Let $\ket{\psi}\in (\mathbb{C}^d)^{\otimes n}$ be an $n$-qudit state with stabilizer fidelity $F\geq \tau$, where $d>2$ is a prime, and $\ket{\phi}=\arg \max_{\ket{\varphi}\in \mathcal{S}_d^n} \left|\braket{\varphi|\psi}\right|^2$. Assume that $H$ is a $\frac{(d-1)\tau^6}{8d^3}$-high-correlation basis. If $\mathrm{span}(H) \neq \mathrm{Weyl}(\ket{\phi})$, then we have
\begin{equation}
\Pr_{\tiny{\begin{smallmatrix}
\mathbf{g}\sim \mathbf{B}_{\psi},\\
s(\mathbf{g})\sim \mathbb{F}_d
\end{smallmatrix}}}\left[\left|\bra{\psi}W_{\mathbf{g}}\ket{\psi}\right|^2 \leq 1-\tfrac{1}{12d^2} \land W_{\mathbf{g}}\ket{\phi}=\omega^{s(\mathbf{g})} \ket{\phi}\right] \geq \frac{(d-1)\tau^6}{8d^4}.
\end{equation}
\end{theorem}

\begin{proof}
Let $S=\mathrm{Weyl}(\ket{\phi})$. Since $\mathrm{span}(H) \neq S$, the difference set $S \setminus \mathrm{span}(H)$ is a nonempty subset of $S$. Hence,
\begin{align}
&\quad \Pr_{\mathbf{g}\sim \mathbf{B}_{\psi}}\left[\left|\bra{\psi}W_{\mathbf{g}}\ket{\psi}\right|^2\leq 1-\tfrac{1}{12d^2} \land \mathbf{g} \in S\right] \nonumber\\
&\geq \Pr_{\mathbf{g}\sim \mathbf{B}_{\psi}}\left[\left|\bra{\psi}W_{\mathbf{g}}\ket{\psi}\right|^2\leq 1-\tfrac{1}{12d^2} \land \mathbf{g}\in S \land \mathbf{g}\notin \mathrm{span}(H)\right] \nonumber\\
&=\Pr_{\mathbf{g}\sim \mathbf{B}_{\psi}}\left[\mathbf{g}\in S \land \mathbf{g}\notin \mathrm{span}(H)\right]-\Pr_{\mathbf{g}\sim \mathbf{B}_{\psi}}\left[\left|\bra{\psi}W_{\mathbf{g}}\ket{\psi}\right|^2 > 1-\tfrac{1}{12d^2} \land \mathbf{g}\in S \land \mathbf{g}\notin \mathrm{span}(H)\right] \\
&\geq \Pr_{\mathbf{g}\sim \mathbf{B}_{\psi}}\left[\mathbf{g}\in S \land \mathbf{g}\notin \mathrm{span}(H)\right] -\Pr_{\mathbf{g}\sim \mathbf{B}_{\psi}}\left[\left|\bra{\psi}W_{\mathbf{g}}\ket{\psi}\right|^2 > 1-\tfrac{1}{12d^2}  \land \mathbf{g}\notin \mathrm{span}(H)\right] \nonumber\\
&\geq \frac{(d-1)\tau^6}{4d^3}-\frac{(d-1)\tau^6}{8d^3} \nonumber\\
&=\frac{(d-1)\tau^6}{8d^3}, \nonumber
\end{align}
where the last inequality follows from Theorem~\ref{theorem B_psi} and the assumption that $H$ is a high-correlation basis. Since there is a probability of $\frac{1}{d}$ to select $s(\mathbf{g}) \in \mathbb{F}_d$ such that $W_{\mathbf{g}}\ket{\phi} = \omega^{s(\mathbf{g})} \ket{\phi}$, the overall lower bound on the probability is $\frac{(d-1)\tau^6}{8d^4}$.
\end{proof}

{\bf Step 4: Bootstrap by measuring.}

Suppose $W_{\mathbf{x}}\ket{\phi}=\omega^{s(\mathbf{x})} \ket{\phi}$, and define the operator $M_{\mathbf{x}} = \frac{I + \omega^{s(\mathbf{x})}W_{\mathbf{x}}^\dagger}{2}$. Clearly, $M_{\mathbf{x}}^\dagger \ket{\phi} = \ket{\phi}$, and if $W_{\mathbf{x}}$ commutes with $W_{\mathbf{y}}$, then $M_{\mathbf{x}}$ commutes with $M_{\mathbf{y}}$ as well. Now consider the POVM $\{M_{\mathbf{x}}^\dagger M_{\mathbf{x}}, I - M_{\mathbf{x}}^\dagger M_{\mathbf{x}}\}$. When this measurement is performed on $\ket{\psi}$ and the outcome corresponding to $M_{\mathbf{x}}^\dagger M_{\mathbf{x}}$ is obtained, the post‑measurement state collapses to $\ket{\psi'} = \frac{M_{\mathbf{x}}\ket{\psi}}{\sqrt{\bra{\psi}M_{\mathbf{x}}^\dagger M_{\mathbf{x}}\ket{\psi}}}$. Next, we prove that performing this measurement increases the fidelity with the state $\ket{\phi}$.

\begin{lemma}\label{lemma increase tau}
Suppose that $W_{\mathbf{x}}\ket{\phi}=\omega^{s(\mathbf{x})} \ket{\phi}$ and $\left|\bra{\psi}W_{\mathbf{x}}\ket{\psi}\right|^2\leq 1-\tfrac{1}{12d^2}$. Let $M_{\mathbf{x}}=\frac{I+\omega^{s(\mathbf{x})}W_{\mathbf{x}}^\dagger}{2}$ and $\ket{\psi'}=\frac{M_{\mathbf{x}}\ket{\psi}}{\sqrt{\bra{\psi}M_{\mathbf{x}}^\dagger M_{\mathbf{x}}\ket{\psi}}}$. Then, we have
\begin{equation}
\left|\braket{\psi'|\phi}\right|^2 \geq \Delta_d \left|\braket{\psi|\phi}\right|^2,
\end{equation}
where $\Delta_d=24d^2\left( 1-\sqrt{1-\frac{1}{12d^2}} \right)$.
\end{lemma}

\begin{proof}
The fidelity between the post-measurement state $\ket{\psi'}$ and the stabilizer state $\ket{\phi}$ is given by $\left|\braket{\psi'|\phi}\right|^2 = \frac{\left|\bra{\psi}M_{\mathbf{x}}^\dagger\ket{\phi}\right|^2}{\bra{\psi}M_{\mathbf{x}}^\dagger M_{\mathbf{x}}\ket{\psi}} = \frac{\left|\braket{\psi|\phi}\right|^2}{\bra{\psi}M_{\mathbf{x}}^\dagger M_{\mathbf{x}}\ket{\psi}}$. We can now bound $\bra{\psi} M_{\mathbf{x}}^\dagger M_{\mathbf{x}} \ket{\psi}$ as follows:
\begin{equation}
\begin{aligned}
\bra{\psi}M_{\mathbf{x}}^\dagger M_{\mathbf{x}}\ket{\psi} &= \frac 1 4 \left( 2+\bra{\psi} \omega^{-s(\mathbf{x})}W_{\mathbf{x}} \ket{\psi}+\bra{\psi} \omega^{s(\mathbf{x})}W_{\mathbf{x}}^\dagger \ket{\psi}  \right)\\
&=\frac 1 2 +\frac 1 2 \mathrm{Re}\left(\bra{\psi} \omega^{-s(\mathbf{x})}W_{\mathbf{x}} \ket{\psi}\right)\\
& \leq \frac 1 2 +\frac 1 2 \left|\bra{\psi} W_{\mathbf{x}} \ket{\psi}\right|.
\end{aligned}
\end{equation}

Thus, 
\begin{equation}
\left|\braket{\psi'|\phi}\right|^2 \geq \frac{2\left|\braket{\psi|\phi}\right|^2}{1+\left|\bra{\psi} W_{\mathbf{x}} \ket{\psi}\right|}\geq \Delta_d \left|\braket{\psi|\phi}\right|^2,
\end{equation}
where $\Delta_d=24d^2\left( 1-\sqrt{1-\frac{1}{12d^2}} \right)$.
\end{proof}

Note that 
\begin{equation}
   1< \Delta_d \leq \Delta_3 = 12(18-\sqrt{321}).
\end{equation}
Hence, the fidelity between the post‑measurement state and the stabilizer state increases by a constant factor. This guarantees that iterating Steps~1-4 converges after a bounded number of rounds, since the fidelity cannot exceed $1$. It remains to verify that a sufficient number of copies of $\ket{\psi'}$ can be prepared from $\ket{\psi}$.

\begin{lemma}[\cite{CGYZ25}]\label{lemma prepare psi_t}
Let $\tau,\delta > 0$, and let $\ket{\psi}, \ket{\phi}$ be $n$-qudit states with fidelity $F(\ket{\psi},\ket{\phi}) \geq \tau$. 
Suppose $\mathfrak{B} = \{M_1^\dagger, M_2^\dagger, \dots, M_t^\dagger\}$ is a set of operators such that each $M_i^\dagger$ stabilizes $\ket{\phi}$. By performing measurements and post-selection, we can, with probability at least $1 - \delta$, prepare $N$ copies of the state $\ket{\psi'} = \frac{M_t \cdots M_1 \ket{\psi}}{\sqrt{\bra{\psi} M_1^\dagger \cdots M_t^\dagger M_t \cdots M_1 \ket{\psi}}}$ using $m_{\mathrm{prepare}} = \frac{2}{\tau} \left( N + \log \frac{1}{\delta} \right)$ copies of $\ket{\psi}$. The running time is $O(m_{\mathrm{prepare}} T t)$, where $T$ is the time for measuring one operator.
\end{lemma}

The complete algorithm for agnostic learning of stabilizer states via $\mathbf{B}_{\psi}$ sampling is presented in Algorithm~\ref{algorithm learn states}. We say that Algorithm~\ref{algorithm learn states} succeeds at the $t$-th iteration if either $\ket{\phi_i} = \ket{\phi}$ for some $0 \le i < t$, or for every $0 \le i < t$ we have $W_{y_i} \ket{\phi} = \omega^{s_i} \ket{\phi}$ and $\bigl| \bra{\psi_i} W_{y_i} \ket{\psi_i} \bigr|^2 \le 1 - \frac{1}{12d^2}$. We denote this event by $B_t$.

\begin{algorithm}[!t]
\caption{Agnostic learning of stabilizer states}
\label{algorithm learn states} 
\KwIn{$\tau>0$, $\Delta_d=24d^2\left( 1-\sqrt{1-\frac{1}{12d^2}} \right)$, copies of an $n$-qudit state $\ket{\psi} \in (\mathbb{C}^d)^{\otimes n}$ where $d>2$ is a prime}
\KwOut{A stabilizer state $\ket{\phi}$ with probability at least $\left(\frac{\tau}{d}\right)^{O(d^2 \log \frac{1}{\tau})}$ satisfying $F(\ket{\psi},\ket{\phi}) \geq \tau$}


Set $\mathfrak{B}_0=\varnothing$, $\mathfrak{R}=\varnothing$,  $t_{\max}=\lfloor \log_{\Delta_d}\frac{1}{\tau} \rfloor$, $\ket{\psi_0}=\ket{\psi}$, $\tau_0=\tau$.

\For{ $t=0$ \KwTo $t_{\max}$} {

Apply Lemma~\ref{lemma prepare psi_t} to prepare $O\left(\frac{d^6n}{\tau_t^{6}}\left( \log n+\log \frac{1}{\tau_t} \right)\right)$ copies of $\ket{\psi_t}$ from $\ket{\psi}$ with $\delta=\frac{1}{9},\mathfrak{B}=\mathfrak{B}_t$. \label{algorithm2 Line3}

Run Algorithm~\ref{algorithm learn H} on $\ket{\psi_t}$ with $\delta=\frac{1}{4},\varepsilon=  \frac{(d-1)\tau_t^6}{8d^3}$. Denote the output by $H_t$.\label{algorithm2 Line4}

Measure $\ket{\psi}$ on the joint eigenbasis of $\mathrm{span}(H_t)$ once. Denote the result $\ket{\phi_t}$. Set $\mathfrak{R} \leftarrow \mathfrak{R}\cup \{\ket{\phi_t}\}$.\label{algorithm2 Line5}

Run $\mathbf{B}_{\psi}$ sampling on $\ket{\psi_t}$ once. Denote the sample by $\mathbf{y}_t$.\label{algorithm2 Line6}

Randomly select $s_t\in\{0,1,\dots,d-1\}$.\label{algorithm2 Line7}

Define $\mathfrak{B}_{t+1}=\mathfrak{B}_t\cup \{M_{(s_t,\mathbf{y}_t)}^\dagger\}$, where $M_{(s_t,\mathbf{y}_t)}=\frac{I+\omega^{s_t}W_{\mathbf{y}_t}^\dagger}{2}$. Define $\tau_{t+1}=\Delta_d\tau_t$ and $\ket{\psi_{t+1}}=\frac{M_{(s_t,\mathbf{y}_t)}\ket{\psi_t}}{\sqrt{\bra{\psi_t}M_{(s_t,\mathbf{y}_t)}^\dagger M_{(s_t,\mathbf{y}_t)} \ket{\psi_t}}}$. \label{algorithm2 Line8}
}

\KwRet{a uniformly random element $\ket{\phi_r}$ from $\mathfrak{R}$.}
\end{algorithm}

\begin{lemma}\label{lemma Pr[B_t]}
If Algorithm~\ref{algorithm learn states} succeeds at the $t$-th iteration, then it will succeed at the $(t+1)$-th iteration with probability at least $\frac{(d-1)\tau^6}{12d^4}$, i.e.,
\begin{equation}
\Pr\left[B_{t+1} | B_t \right] \geq \frac{(d-1)\tau^6}{12d^4}.
\end{equation}
\end{lemma}

\begin{proof}
On the event $B_t$, one of two cases occurs: either $\ket{\phi_i} = \ket{\phi}$ for some $0 \le i < t$, which immediately implies $B_{t+1}$; or $W_{y_i} \ket{\phi} = \omega^{s_i} \ket{\phi}$ and $\bigl| \bra{\psi_i} W_{y_i} \ket{\psi_i} \bigr|^2 \le 1 - \frac{1}{12d^2}$ for all $0 \le i < t$, in which case the fidelity satisfies $\bigl| \braket{\psi_t \mid \phi} \bigr|^2 \ge \Delta_d^t \tau = \tau_t$. Consequently, at the $t$-th iteration, Line~\ref{algorithm2 Line3} of Algorithm~\ref{algorithm learn states} succeeds in preparing the required copies of $\ket{\psi_t}$ with probability at least $\frac{8}{9}$.
Moreover, by Theorem~\ref{theorem learn H}, Line~\ref{algorithm2 Line4} outputs a $\frac{(d-1)\tau_t^6}{8d^3}$-high-correlation basis $H_t$ with probability at least $\frac{3}{4}$.

If $H_t = \mathrm{Weyl}(\ket{\phi})$, then we obtain $\ket{\phi_t} = \ket{\phi}$ with probability at least $\tau$.  
If $H_t \neq \mathrm{Weyl}(\ket{\phi})$, then by Theorem~\ref{theorem lower bound low correlation}, with probability at least $\frac{(d-1)\tau_t^6}{8d^4}$ we have
$ \bigl| \bra{\psi} W_{\mathbf{y}} \ket{\psi} \bigr|^2 \le 1 - \frac{1}{12d^2}$ and $W_{\mathbf{y}} \ket{\phi} = \omega^{s(\mathbf{y})} \ket{\phi}$. Therefore,
\begin{equation}
\Pr\left[B_t|B_{t-1}\right] \geq \min\left\{1,\frac89 \cdot \frac34 \cdot \min\left\{\tau,\frac{(d-1)\tau_t^6}{8d^4}\right\}\right\} \geq \frac{(d-1)\tau^6}{12d^4}.
\end{equation}
\end{proof}

Intuitively, each successful round either identifies the target state $\ket{\phi}$ explicitly or strictly boosts the fidelity towards it. Consequently, the algorithm must, with high probability, halt after a bounded number of steps and output $\ket{\phi}$. This convergence argument underpins our main result, which we will proceed to prove rigorously.

\begin{theorem}\label{theorem find phi}
Let $\ket{\psi} \in (\mathbb{C}^d)^{\otimes n}$ be an unknown $n$-qudit quantum state, where $d>2$ is a prime. Suppose $\ket{\phi}$ is a stabilizer state satisfying $F(\ket{\psi},\ket{\phi}) \geq \tau$. Given copies of $\ket{\psi}$, there exists a quantum algorithm that outputs $\ket{\phi}$ with probability at least $\left(\frac{\tau}{d}\right)^{O(d^2 \log \frac{1}{\tau})}$. This algorithm uses only single-copy and four-copy measurements with  $O\left( \frac{d^6n}{\tau^{7}} \left( \log n+\log \frac{1}{\tau} \right) \right)$ samples and $O\left( \frac{d^2n^2}{\tau^{6}} \left( n+\frac{d^4}{\tau}\left( \log n+\log \frac{1}{\tau} \right) \right) \right)$ time.
\end{theorem}

\begin{proof}
Since $B_0$ holds trivially, we have $\Pr[B_0]=1$. By Lemma~\ref{lemma Pr[B_t]}, it follows that
\begin{equation}
\Pr[B_{t_{\max}+1}] \geq \left( \frac{(d-1)\tau^6}{12d^4} \right)^{t_{\max}+1}.
\end{equation}

Note that the event $B_{t_{\max}+1}$ implies either $\ket{\phi} \in \mathfrak{R}$ or $\left|\bra{\psi}W_{y_i}\ket{\psi}\right|^2 \leq 1-\tfrac{1}{12d^2} \land W_{y_i}\ket{\phi}=\omega^{s(\mathbf{y}_i)}\ket{\phi}$ for all $0\leq i \leq t_{\max}$. If the latter holds, then by Lemma~\ref{lemma increase tau}, we have $\left|\braket{\psi_{t_{\max}+1}|\phi}\right|^2 \geq \Delta_d^{t_{\max}+1}\tau >1$, which is impossible. Therefore, when $B_{t_{\max}+1}$ occurs, we must have $\ket{\phi} \in \mathfrak{R}$. Since $|\mathfrak{R}|\leq t_{\max}+1$, it holds that
\begin{equation}
\Pr\left[\ket{\phi_r}=\ket{\phi}\right] \geq \frac{\Pr\left[B_{t_{\max}+1}\right]}{t_{\max}+1}\geq \frac{1}{1+\log_{\Delta_d}\frac{1}{\tau}}\left( \frac{(d-1)\tau^6}{12d^4} \right)^{1+\log_{\Delta_d}\frac{1}{\tau}}.
\end{equation}
Note that $\Delta_d = 1+\Theta(1/d^2)$. Therefore, we have
\begin{equation}
\Pr\left[\ket{\phi_r} = \ket{\phi}\right] \geq \left(\frac{\tau}{d}\right)^{O(d^2 \log \frac{1}{\tau})}.
\end{equation}

Now, we analyze the sample complexity and time complexity of Algorithm~\ref{algorithm learn states}. Line~\ref{algorithm2 Line3} requires $\frac{2}{\tau}\left( O\left( \frac{d^6n}{\tau_t^{6}} \left( \log n+\log \frac{1}{\tau_t} \right) \right) + 3 \right)$ copies of $\ket{\psi}$. Line~\ref{algorithm2 Line5} requires one copy of $\ket{\psi}$, and Line~\ref{algorithm2 Line6} requires 8 copies of $\ket{\psi_t}$. Therefore, the total number of samples is
\begin{equation}
\sum_{t=0}^{t_{\max}} 1+\frac{2}{\tau}\left( O\left( \frac{d^6n}{\tau_t^{6}} \left( \log n+\log \frac{1}{\tau_t} \right) \right) + 3 \right) =O\left( \frac{d^6n}{\tau^{7}} \left( \log n+\log \frac{1}{\tau} \right) \right).
\end{equation}

The running time of Line~\ref{algorithm2 Line3} is $O\left( \frac{d^6 t n^2}{\tau \tau_t^{6}} \left( \log n+\log \frac{1}{\tau_t} \right) \right)$, where measuring one operator takes $O(n)$ time. From Theorem~\ref{theorem learn H}, the time complexity of Line~\ref{algorithm2 Line4} is $O\left( \frac{d^2 n^2}{\tau_t^{6}} \left(n+d^4\log\frac{1}{\tau_t}\right) \right)$. According to Theorem~\ref{theorem circuit}, the gate complexity of the Clifford circuit used for measurement in the basis of $\mathrm{span}(H_t)$ in Line~\ref{algorithm2 Line5} is $O(n^2)$, so the time required for Line~\ref{algorithm2 Line5} is $O(n^2)$. A single $\mathbf{B}_{\psi}$ sampling in Line~\ref{algorithm2 Line6} takes $O(n)$ time, and both Line~\ref{algorithm2 Line7} and Line~\ref{algorithm2 Line8} require $O(1)$ time. Therefore, the total time complexity is
\begin{equation}
\sum_{t=0}^{t_{\max}}  O\biggl( \frac{d^6 t n^2}{\tau \tau_t^{6}} \Bigl( \log n + \log \frac{1}{\tau_t} \Bigr) \biggr)  + \sum_{t=0}^{t_{\max}} O\biggl( \frac{d^2 n^2}{\tau_t^{6}} \Bigl( n + d^4 \log \frac{1}{\tau_t} \Bigr) \biggr)  + O(n^2)= O\biggl( \frac{d^2 n^2}{\tau^{6}} \biggl( n + \frac{d^4}{\tau} \Bigl( \log n + \log \frac{1}{\tau} \Bigr) \biggr) \biggr).
\end{equation}
\end{proof}

By repeating the algorithm in Theorem~\ref{theorem find phi} sufficiently many times, we immediately obtain a list of stabilizer states of length $O\left(\log \frac{1}{\delta}\right) \cdot \left(\frac{d}{\tau}\right)^{O(d^2 \log \frac{1}{\tau})}$ containing all states with fidelity at least $\tau$ with $\ket{\psi}$.

\begin{corollary}[List-decoding stabilizer states]\label{list-decoding}
Let $\tau,\delta > 0$ and $N = O\left(\log \frac{1}{\delta}\right) \cdot \left(\frac{d}{\tau}\right)^{O(d^2 \log \frac{1}{\tau})}$. Suppose $\ket{\psi} \in (\mathbb{C}^d)^{\otimes n}$ is an unknown $n$-qudit quantum state, where $d>2$ is a prime. Given copies of $\ket{\psi}$, there exists a quantum algorithm that, with probability at least $1 - \delta$, outputs $L = \{ \ket{\phi_1}, \dots, \ket{\phi_N} \}$ a list of stabilizer states such that every stabilizer state with fidelity at least $\tau$ to $\ket{\psi}$ is contained in $L$. The algorithm uses $O\left(n \log n \log \frac{1}{\delta}\right) \cdot \left(\frac{d}{\tau}\right)^{O(d^2 \log \frac{1}{\tau})}$ copies and $O\left(n^3 \log \frac{1}{\delta} \right) \cdot \left(\frac{d}{\tau}\right)^{O(d^2 \log \frac{1}{\tau})}$ time.
\end{corollary}

\begin{proof}
Let $S' = \{ \ket{\phi} \in \mathcal{S}_d^n : F(\ket{\psi}, \ket{\phi}) \geq \tau \}$. By Theorem~\ref{theorem find phi}, for any fixed $\ket{\phi} \in S'$, a single run of Algorithm~\ref{algorithm learn states} on copies of $\ket{\psi}$ outputs $\ket{\phi}$ with probability at least $p := \bigl(\frac{\tau}{d}\bigr)^{O(d^2 \log \frac{1}{\tau})}$. Because the outputs corresponding to distinct states in $S'$ are mutually exclusive events, their probabilities must sum to at most $1$. Consequently, $|S'| \le 1/p$, and we set $M = 1/p$.

We run Algorithm~\ref{algorithm learn states} independently $N$ times, collect the resulting output states, and form the list $L = \{\ket{\phi_1}, \dots, \ket{\phi_N}\}$. Then, for any fixed state $\ket{\phi} \in S'$, the probability that $\ket{\phi}$ does not appear in $L$ satisfies
\begin{equation}
\Pr[\,\ket{\phi} \notin L\,] \leq (1-p)^N \leq \mathrm{e}^{-pN}.
\end{equation}

To ensure $S' \subseteq L$, it suffices that every state $\ket{\phi} \in S'$ is included in the list $L$. Applying the union bound, we obtain
\begin{equation}
\Pr[\,S' \not\subseteq L\,] \le \sum_{\ket{\phi} \in S'} \Pr[\,\ket{\phi} \notin L\,]  \le M \mathrm{e}^{-pN}.
\end{equation}
Choosing $N = \bigl\lceil \tfrac{1}{p} \log \tfrac{M}{\delta} \bigr\rceil = O\bigl(\log \tfrac{1}{\delta}\bigr) \cdot \bigl(\tfrac{d}{\tau}\bigr)^{O(d^2 \log \frac{1}{\tau})}$ then guarantees $\Pr[\,S' \not\subseteq L\,] \le \delta$, as required.

Therefore, with probability at least $1-\delta$, running Algorithm~\ref{algorithm learn states} $N$ times yields a list that contains every stabilizer state satisfying the fidelity condition. The required sample complexity is $O\left( N \frac{d^6n}{\tau^{7}}\left( \log n+\log \frac{1}{\tau} \right) \right) =O\left(n \log n \log \frac{1}{\delta}\right)\cdot \left(\frac{d}{\tau}\right)^{O(d^2 \log \frac{1}{\tau})}$, and the runtime is 
\begin{equation}
O\left( N \frac{d^2 n^2}{\tau^{6}}\left( n+\frac{d^4}{\tau}\left( \log n+\log \frac{1}{\tau} \right) \right) \right)=O\left(n^3 \log \frac{1}{\delta} \right) \cdot \left(\frac{d}{\tau}\right)^{O(d^2 \log \frac{1}{\tau})}.
\end{equation}
\end{proof}

Using the classical shadows described in Theorem~\ref{classical shadows}, we estimate the fidelity between $\ket{\psi}$ and every candidate state in the list produced by Corollary~\ref{list-decoding}. The state achieving the highest estimated fidelity is then selected, and it is precisely the target stabilizer state $\ket{\phi}$. Combining the sample and time costs of this fidelity estimation step with those of the list‑decoding procedure yields our final complexity bounds, thereby completing the proof.

\begin{theorem}[Agnostic learning of stabilizer states] \label{theorem learn states finally}
Let $\tau\geq \varepsilon>0$ and $\delta > 0$. Let $\ket{\psi} \in (\mathbb{C}^d)^{\otimes n}$ be an unknown $n$-qudit quantum state with stabilizer fidelity $F_{\mathcal{S}}(\ket{\psi}) \geq \tau$, where $d>2$ is a prime. Given copies of $\ket{\psi}$, there exists a quantum algorithm that outputs a stabilizer state $\ket{\phi}$ with probability at least $1 - \delta$ such that $F(\ket{\psi}, \ket{\phi}) \geq F_{\mathcal{S}}(\ket{\psi}) - \varepsilon$. The algorithm uses only single-copy and four-copy measurements with $O\left(n \log n \log \frac{1}{\delta}\right) \cdot \left(\frac{d}{\tau}\right)^{O(d^2 \log \frac{1}{\tau})}+O\left( \frac{d}{\varepsilon^2}\left( \log \frac{1}{\delta} + \log^2\frac{1}{\tau} \right) \right)$ samples and $O\left(\frac{n^3}{\varepsilon^2}\log^2 \frac{1}{\delta} \right) \cdot \left(\frac{d}{\tau}\right)^{O(d^2 \log \frac{1}{\tau})}$ time.
\end{theorem}

Specifically, in the regime where $\tau$ is constant, the algorithm is efficient in both sample and time complexity. When $1/\tau = \mathrm{poly}(n)$, the sample complexity and the runtime become quasi-polynomial in $n$, matching the complexity of the stabilizer bootstrapping algorithm for $n$-qubit quantum states.

Furthermore, Theorem~\ref{theorem learn states finally} enables us to estimate the $n$-qudit stabilizer fidelity, from which the following conclusion follows naturally.

\begin{corollary}[Efficient estimation of stabilizer fidelity]
Let $\varepsilon, \delta > 0$ and $\ket{\psi} \in (\mathbb{C}^d)^{\otimes n}$ an unknown $n$-qudit quantum state, where $d>2$ is a prime. Given copies of $\ket{\psi}$, there exists a quantum algorithm that estimates $F_{\mathcal{S}}(\ket{\psi})$ within additive error $\varepsilon$ with probability at least $1 - \delta$. The algorithm uses only single-copy and four-copy measurements with $O\left(n \log n \log \frac{1}{\delta}\right) \cdot \left(\frac{d}{\varepsilon}\right)^{O(d^2 \log \frac{1}{\varepsilon})}$ samples and $O\left(n^3\log^2 \frac{1}{\delta} \right) \cdot \left(\frac{d}{\varepsilon}\right)^{O(d^2 \log \frac{1}{\varepsilon})}$ time.
\end{corollary}

\section{Agnostic learning in the high-fidelity regime} \label{bounded-distance section}

We find that the stabilizer bootstrapping framework simplifies to a more streamlined algorithm when the unknown quantum state has a large stabilizer fidelity. The starting point is the difference in the correlations of Weyl operators with $\ket{\psi}$ depending on whether the operators are in $\mathrm{Weyl}(\ket{\phi})$ or outside $\mathrm{Weyl}(\ket{\phi})$.

\begin{proposition}[\cite{ADIS25}]\label{correlation lower bound}
    Let $\ket{\psi} \in (\mathbb{C}^d)^{\otimes n}$ be an $n$-qudit quantum state that has fidelity $\tau$ with a stabilizer state $\ket{\phi}$, where $\tau \ge 1/2$. If $\mathbf{x} \in \mathrm{Weyl}(\ket{\phi})$, then $\left|\bra{\psi}W_{\mathbf{x}}\ket{\psi}\right|^2 \ge (2\tau - 1)^2$.
\end{proposition}

\begin{proposition}\label{correlation upper bound}
    Let $\ket{\psi} \in (\mathbb{C}^d)^{\otimes n}$ be an $n$-qudit quantum state that has fidelity $\tau$ with a stabilizer state $\ket{\phi}$, where $\tau \ge 1/2$. If $\mathbf{y} \notin \mathrm{Weyl}(\ket{\phi})$, then 
    \begin{equation}
        \left|\bra{\psi}W_{\mathbf{y}}\ket{\psi}\right|^2 \le 4\tau(1-\tau).
    \end{equation}
\end{proposition}

\begin{proof}
    Since $\mathbf{y} \notin \mathrm{Weyl}(\ket{\phi})$, there exists an operator $W_{\mathbf{x}} \in \mathrm{Weyl}(\ket{\phi})$ that does not commute with $W_{\mathbf{y}}$, meaning $W_{\mathbf{x}} W_{\mathbf{y}} = \omega^c W_{\mathbf{y}} W_{\mathbf{x}}$ for some integer $c \not\equiv 0 \pmod d$. Because $W_{\mathbf{x}}$ stabilizes $\ket{\phi}$, we have
    \begin{equation}
        \bra{\phi}W_{\mathbf{y}}\ket{\phi} = \bra{\phi}W_{\mathbf{x}}^\dagger W_{\mathbf{x}} W_{\mathbf{y}}\ket{\phi} = \omega^c \bra{\phi}W_{\mathbf{y}}\ket{\phi}.
    \end{equation}
    Since $\omega^c \neq 1$, it strictly follows that $\bra{\phi}W_{\mathbf{y}}\ket{\phi} = 0$.
    
    Now, we decompose the input state as $\ket{\psi} = \sqrt{\tau}\ket{\phi} + \sqrt{1-\tau}\ket{\phi^\perp}$, where $\ket{\phi^\perp}$ is a normalized state orthogonal to $\ket{\phi}$. Then, we have
    \begin{equation}
        \begin{aligned}
        \bra{\psi}W_{\mathbf{y}}\ket{\psi}
        &= \left(\sqrt{\tau}\bra{\phi} + \sqrt{1-\tau}\bra{\phi^\perp}\right) W_{\mathbf{y}} \left(\sqrt{\tau}\ket{\phi} + \sqrt{1-\tau}\ket{\phi^\perp}\right) \\
        &= \sqrt{\tau(1-\tau)}\bra{\phi}W_{\mathbf{y}}\ket{\phi^\perp} + \sqrt{\tau(1-\tau)}\bra{\phi^\perp}W_{\mathbf{y}}\ket{\phi}+ (1-\tau)\bra{\phi^\perp}W_{\mathbf{y}}\ket{\phi^\perp}.
    \end{aligned}
    \end{equation}

    Let $a = \bra{\phi}W_{\mathbf{y}}\ket{\phi^\perp}$, and we can get $ W_{\mathbf{y}}\ket{\phi^\perp} = a \ket{\phi} + \ket{r}$, where $\ket{r}$ is a non-normalized state orthogonal to $\ket{\phi}$ with $\|\ket{r}\| = \sqrt{1-|a|^2}$. For ease of description, let $\ket{u} = W_{\mathbf{y}} \ket{\phi}$, and we have 
    \begin{equation}
        0 = \langle \phi | \phi^\perp\rangle = \bra{\phi}W_{\mathbf{y}}^\dagger W_{\mathbf{y}} \ket{\phi^\perp} = a \langle u | \phi \rangle + \langle u | r\rangle = \langle u | r\rangle.
    \end{equation}
    Then, we can bound $\left|\bra{\psi}W_{\mathbf{y}}\ket{\psi}\right|$ as follows:
    \begin{equation}
    \begin{aligned}
         \left|\bra{\psi}W_{\mathbf{y}}\ket{\psi}\right|
        &= \left|\sqrt{\tau(1-\tau)} a + \sqrt{\tau(1-\tau)}\braket{\phi^\perp | u} + (1-\tau)\braket{\phi^\perp | r}\right| \\
        &\le \sqrt{\tau(1-\tau)} |a| + \left|\bra{\phi^\perp}(\sqrt{\tau(1-\tau)}\ket{u} + (1-\tau)\ket{r})\right| \\
        &\le \sqrt{\tau(1-\tau)} |a| + \|\sqrt{\tau(1-\tau)}\ket{u} + (1-\tau)\ket{r}\|_2 \\
        &= \sqrt{\tau(1-\tau)} |a| + \sqrt{\tau(1-\tau) + (1-\tau)^2(1-|a|^2)}.
    \end{aligned}
    \end{equation}
    
    Let $f(z) = \sqrt{\tau(1-\tau)} z + \sqrt{\tau(1-\tau) + (1-\tau)^2(1-z^2)}$. For any $\tau \ge 1/2$, one easily checks that $\max_{z \in [0,1]} f(z) = f(1) = 2\sqrt{\tau(1-\tau)}$. Therefore,
    \begin{equation}
    \left|\bra{\psi}W_{\mathbf{y}}\ket{\psi}\right|^2 \le 4\tau(1-\tau).
    \end{equation}
\end{proof}

Proposition~\ref{correlation lower bound} and Proposition~\ref{correlation upper bound} suggest that we can determine whether a given Pauli operator is in the unsigned stabilizer group only from its correlation with $\ket{\psi}$, provided that for all $\mathbf{x} \in \mathrm{Weyl}(\ket{\phi})$ and for all $\mathbf{y} \notin \mathrm{Weyl}(\ket{\phi})$, $\left|\bra{\psi}W_{\mathbf{y}}\ket{\psi}\right|^2 < \left|\bra{\psi}W_{\mathbf{x}}\ket{\psi}\right|^2$. This condition holds exactly when $4\tau(1-\tau)<(2\tau-1)^2$, i.e., $\tau>\cos^2(\pi/8)$. However, we must also account for the fact that we cannot know the correlations exactly. Instead, we can only estimate them to within $\pm O(\gamma)$ accuracy. Consequently, $\tau$ must be at least $\cos^2(\pi/8)+\gamma$ for some $\gamma>0$. We formalize this in the following corollary.

\begin{corollary}\label{corollary filtering}
    Let $\ket{\psi}$ be an $n$-qudit quantum state that has fidelity $\cos^2(\pi/8) + \gamma$ with a stabilizer state $\ket{\phi}$ for some $\gamma > 0$. Then for all $\mathbf{x} \in \mathrm{Weyl}(\ket{\phi})$ and all $\mathbf{y} \notin \mathrm{Weyl}(\ket{\phi})$,
    \begin{equation}
        \left|\bra{\psi}W_{\mathbf{x}}\ket{\psi}\right|^2  > \frac{1}{2}+2\sqrt{2}\gamma, ~ \text{and} ~ \left|\bra{\psi}W_{\mathbf{y}}\ket{\psi}\right|^2  < \frac{1}{2}-2\sqrt{2}\gamma.
    \end{equation}
\end{corollary}

A noteworthy consequence of Corollary~\ref{corollary filtering} is that the stabilizer state $\ket{\phi}$ is unique.

\begin{corollary}[\cite{GIKL24b}]
    If $\ket{\psi}$ has fidelity at least $\cos^2(\pi/8) + \gamma$ with a stabilizer state $\ket{\phi}$ for some $\gamma > 0$, then $\ket{\phi}$ must be unique.
\end{corollary}

The core idea of the learning algorithm is simple. First, we use skewed Bell difference sampling to collect a large number of Weyl operators that are guaranteed to contain a generating set for $\mathrm{Weyl}(\ket{\phi})$. Next, we estimate their correlations with $\ket{\psi}$ via the SWAP test and discard those operators that lie outside $\mathrm{Weyl}(\ket{\phi})$, i.e., those satisfying $\bigl| \bra{\psi} W_{\mathbf{y}} \ket{\psi} \bigr|^2 < 1/2$. The remaining Weyl operators then enable us to learn $\ket{\phi}$ with high probability. The full procedure is detailed in Algorithm~\ref{bounded-distance algorithm}.

\begin{algorithm}[!t]
\SetKwInOut{KwPromise}{Promise}
\caption{Bounded-distance agnostic learning of stabilizer states}
\label{bounded-distance algorithm} 
\KwIn{$\delta,\gamma >0$, copies of an $n$-qudit state $\ket{\psi} \in (\mathbb{C}^d)^{\otimes n}$ where $d>2$ is a prime}
\KwPromise{$\ket{\psi}$ has fidelity at least $\cos^2(\pi/8) + \gamma$ with a stabilizer state $\ket{\phi}$} 
\KwOut{$\ket{\phi}$ with probability at least $1-\delta$}

Let $m = \frac{8d^3}{(d-1)\cos^{12}(\pi/8)}(n + \log(3/\delta))$.

Draw $m$ samples $\mathbf{x}_1, \ldots, \mathbf{x}_m \in \mathbb{F}_d^{2n}$ from the distribution $\mathbf{B}_\psi$.

Using the algorithm in Theorem~\ref{theorem M estimations}, estimate $\left|\bra{\psi}W_{\mathbf{x}_i}\ket{\psi}\right|^2$ for all $i$ to accuracy $\pm2\sqrt{2}\gamma$ with failure probability at most $\delta/3$.


Retain $\mathbf{x}_i$'s which the estimate of $\left|\bra{\psi}W_{\mathbf{x}_i}\ket{\psi}\right|^2 $ is greater than $1/2$. Let $S^*$ be the subspace spanned by these samples. If $S^*$ is not Lagrangian, then output ``\textsc{failure}". Otherwise, find a Clifford circuit that measures in the stabilizer basis induced by $S^*$.

Measure $O(\log(1/\delta))$ copies of $\ket{\psi}$ in the stabilizer basis induced by $S^*$ and output the majority result.

\end{algorithm}

We first argue that, with high probability, polynomially many samples drawn from $\mathbf{B}_{\psi}$ suffice to yield a complete generating set for $\mathrm{Weyl}(\ket{\phi})$.

\begin{lemma}\label{lemma sample a complete generators set}
Let $\ket{\psi} \in (\mathbb{C}^d)^{\otimes n}$ be an $n$-qudit state, where $d>2$ is a prime. 
Let $\ket{\phi}$ be a stabilizer state that maximizes the stabilizer fidelity $F_{\mathcal{S}}(\ket{\psi}) \geq \tau$, and set $S = \mathrm{Weyl}(\ket{\phi})$. If we draw $m \ge \frac{8d^3}{(d-1)\tau^6}(n + \log(1/\delta))$ independent samples from $\mathbf{B}_{\psi}$, then, with probability at least $1-\delta$, the samples contain a complete generating set for $S$.
\end{lemma}

\begin{proof}

Let $\mathbf{x}_1, \dots, \mathbf{x}_m \in \mathbb{F}_d^{2n}$ be the outcomes of $m$ independent samples from $\mathbf{B}_{\psi}$. For each $i$, denote by $T_i$ the subspace of $S$ spanned by the elements in $\{\mathbf{x}_1, \dots, \mathbf{x}_i\} \cap S$. Define indicator random variables $X_i$ as
\begin{equation}
X_i = 
\begin{cases}
1, & \text{if } \mathbf{x}_i \in S \setminus T_{i-1} \text{ or } T_{i-1} = S; \\
0, & \text{otherwise}.
\end{cases}
\end{equation}
We aim to prove that $\sum_{i=1}^m X_i \ge n$ holds with high probability, which guarantees that $T_m = S$.

From Theorem~\ref{theorem B_psi}, for any assignment of $X_1,\dots, X_{i-1}$, we have $\mathbf{E}[X_i|X_1,\dots, X_{i-1}] \ge c$ with $c = \frac{d-1}{4d^3}\tau^6$. Note that the $X_i$'s are not independent. To circumvent this difficulty, let $X'_i$ $(i=1,2,\dots,m)$ be $m$ i.i.d. samples from a Bernoulli distribution with $\Pr \left[ X'_i=1 \right]=c$. Then, it is easy to see that 
\begin{equation}
    \Pr\left[X_1+ \cdots+ X_i <n\right] \leq \Pr\left[X'_1+ \cdots+ X'_i <n\right],
\end{equation}
since $\Pr\left[X'_i|X'_1,\dots, X'_{i-1}\right] \leq \Pr\left[X_i | X_1,\dots, X_{i-1}\right]$.

Let $\eta = 1-\frac{n}{cm}$. Then, by a Chernoff Bound, we have 
    \begin{equation}
    \begin{aligned}
        \Pr\left[\sum_{i=1}^m{X'_i} < n\right] &= \Pr\left[\sum_{i=1}^m{X'_i} < (1-\eta)cm\right] \\
        &\le \exp(-\eta^2cm/2) \\
        &= \exp\left(-\frac{cm}{2}\left(1-\frac{2n}{cm}+\frac{n^2}{c^2m^2}\right)\right) \\
        &\le \exp\left(-\frac{cm}{2}\left(1-\frac{2n}{cm}\right)\right)\\
        &= \exp(n-cm/2).
    \end{aligned}
    \end{equation}

Hence, choosing
    \begin{equation}
        m \ge \frac{2}{c}(n+\log(1/\delta)) = \frac{8d^3}{(d-1)\tau^6}(n + \log(1/\delta)),
    \end{equation}
    we can ensure that $\Pr\left[\sum_{i=1}^m{X_i} < n\right] \le \delta$.
\end{proof}

With these preparations in place, the proof of correctness for Algorithm~\ref{bounded-distance algorithm} is now straightforward.

\begin{theorem}
    Let $d>2$ be a prime, and $\ket{\psi} \in (\mathbb{C}^d)^{\otimes n}$ an $n$-qudit state with fidelity at least $\cos^2(\pi/8) + \gamma$ with a stabilizer state $\ket{\phi}$ for $\gamma > 0$. Given $O\left( \frac{d^2}{\gamma^2} \left(n + \log \frac{1}{\delta}\right) \log \frac{n}{\delta} \right)$ copies of $\ket{\psi}$ and $O\left( n^3+ \frac{d^2n}{\gamma^2} \left(n + \log \frac{1}{\delta}\right) \log \frac{n}{\delta} \right)$ time, Algorithm~\ref{bounded-distance algorithm} outputs $\ket{\phi}$ with probability at least $1-\delta$.
\end{theorem}

\begin{proof}
    To prove the correctness and complexity of Algorithm~\ref{bounded-distance algorithm}, we analyze its three main phases, allocating a failure probability of $\delta/3$ to each to bound the overall error. In the first phase, the algorithm draws $m = \frac{8d^3}{(d-1)\cos^{12}(\pi/8)}\left(n + \log\frac{3}{\delta}\right)$ samples from the distribution $\mathbf{B}_\psi$. According to Lemma~\ref{lemma sample a complete generators set}, this number of samples guarantees, with probability at least $1-\delta/3$, that a complete generating set of $S$ appears among the outcomes.

    In the stage of estimating $\left|\bra{\psi}W_{\mathbf{x}_i}\ket{\psi}\right|^2$ for each of the $m$ samples to within $\pm 2\sqrt{2}\gamma$ with failure probability at most $\delta/3$, the overall computational cost is $O\left( \frac{m}{\gamma^2} \log \frac{3m}{\delta} \right)$ copies and $O\left( \frac{mn}{\gamma^2} \log \frac{3m}{\delta} \right)$ time by Theorem~\ref{theorem M estimations}.

    Assuming the estimation phase succeeds, the filtering step correctly compiles the elements of $\mathrm{Weyl}(\ket{\phi})$, according to Corollary~\ref{corollary filtering}. And the span of these elements exactly equals $\mathrm{Weyl}(\ket{\phi})$. In the final phase, determining whether $S^*$ is Lagrangian and computing a Clifford circuit that measures in the basis induced by $S^*$ can be done using Theorem~\ref{theorem circuit} in $O(n^3)$ time. When a copy of $\ket{\psi}$ is measured in the stabilizer basis specified by $\mathrm{Weyl}(\ket{\phi})$, the probability of obtaining outcome $\ket{\phi}$ is at least $\cos^2(\pi/8) \approx 0.853 > 1/2$. By taking the majority vote over $O(\log(1/\delta))$ measurements, we can ensure that the correct state $\ket{\phi}$ is output with failure probability at most $\delta/3$.
    
     By the union bound, the overall probability of failure across the sampling, estimation, and measurement phases is at most $\delta$. Combining the sample and time complexities of the estimation phase with the $O(n^3)$ overhead of the final phase yields the overall complexity stated in the theorem.
\end{proof}

\section{Conclusion}\label{conclusion}

In this work, we established an efficient agnostic learning algorithm for stabilizer states on $n$-qudit systems of odd prime local dimension~$d$, using copies of an unknown pure state~$\ket{\psi}$. Given a lower bound $\tau$ on the stabilizer fidelity $F_{\mathcal{S}}(\ket{\psi})$ and a target accuracy $\varepsilon > 0$, the algorithm outputs, with high probability, a stabilizer state whose fidelity with $\ket{\psi}$ is at least $F_{\mathcal{S}}(\ket{\psi}) - \varepsilon$, achieving sample and time complexities of $(d/\tau)^{O(d^2 \log(1/\tau))} \cdot \mathrm{poly}(n, 1/\varepsilon)$. This is accomplished by extending the stabilizer bootstrapping framework~\cite{CGYZ25} to the qudit setting via skewed Bell difference sampling~\cite{ADIS25}, establishing an anticoncentration property for its marginal distribution, and addressing the non-Hermiticity of qudit Weyl operators in both correlation estimation and the high-fidelity analysis. As a direct corollary, this enables efficient estimation of the stabilizer fidelity, a natural measure of quantum magic. When the stabilizer fidelity of the input state is sufficiently large, the bootstrapping machinery can be bypassed entirely: a single round of sampling and correlation estimation suffices to identify the stabilizer group, yielding polynomial sample and time complexities in all parameters.

There are several natural directions for further investigation.

\begin{enumerate}
    \item Our results are restricted to odd prime local dimension~$d$. Extending them to composite dimensions remains open; a key obstacle is that the output distribution of skewed Bell difference sampling exhibits a more complex structure in this setting, and the techniques used here to establish anticoncentration do not directly carry over.
    \item In the qubit setting, Chen et~al.~\cite{CGYZ25} extended stabilizer bootstrapping to unknown mixed states. The qudit generalization requires establishing an anticoncentration bound in this broader setting, which poses new challenges: the output distribution of skewed Bell difference sampling has yet to be characterized for mixed-state inputs, and the qubit anticoncentration proof~\cite{CGYZ25} relies on the Hermiticity of Pauli operators in ways that do not directly generalize.
    \item Can the quasi-polynomial dependence on $1/\tau$ be improved to polynomial, while retaining polynomial dependence on $n$ and $1/\varepsilon$, for agnostic learning of stabilizer states? In the qubit pure-state setting, Grewal et~al.~\cite{GIKL24b} showed that polynomial sample complexity suffices, but the corresponding time complexity remains exponential. Closing this gap is a major open challenge.
    \item Our algorithm relies on multi-copy joint measurements through skewed Bell difference sampling. In the realizable setting, qubit stabilizer states can be learned from single-copy measurements alone~\cite{CLL24, GIKL25}. Whether single-copy measurements suffice for agnostic learning of stabilizer states remains an intriguing question.
    \item Our algorithm uses only the marginal distribution of a single output string from skewed Bell difference sampling, discarding the correlations between the two strings. When the input is a stabilizer state, the two strings are independent; for a general state, they need not be. Whether this correlation structure can be exploited to extract additional information about the magic of the input state is an interesting question.
\end{enumerate}

\appendices
\section{Clifford circuit synthesis}\label{Clifford circuit synthesis}

Recall that a single-qudit Pauli operator $P$ can be expressed, up to an overall phase, as $P = X^x Z^z$ for some $x, z \in \mathbb{F}_d$. Consequently, any tensor product of $m$ such operators, $P_1 \otimes \cdots \otimes P_m$, admits the analogous factorized form $X^{x_1} Z^{z_1} \otimes \cdots \otimes X^{x_m} Z^{z_m}$. This representation can be conveniently encoded in a compact row vector, often referred to as the \textit{symplectic representation}~\cite{Got24}, namely
\begin{equation}
(x_1~\cdots~x_m~|~z_1~\cdots~z_m),
\end{equation}
where the two blocks correspond to the exponents of the $X$ and $Z$ operators, respectively. For instance, for any $x, z, x' \in \mathbb{F}_d$, we have the correspondences
\begin{equation}
X^xZ^z \leftrightarrow (x~|~z),~~ X^x \otimes I \otimes X^{x'}Z^z \leftrightarrow (x~0~x'|0~0~z).
\end{equation}

Collecting the symplectic representations of multiple Pauli operators yields a rectangular array known as a \textit{stabilizer tableau}~\cite{AG04,DBN13,KNKK25,MYZ25}. Operations on Pauli operators translate naturally into simple manipulations of the corresponding tableau. For example, the symplectic representation of the product of two Pauli operators is simply the component‑wise sum of their individual symplectic vectors.

We now describe how Clifford gates act on the stabilizer tableau via conjugation defined in Eq.~\eqref{eq. conjugation}. Concretely, we write $(x, z) \xrightarrow{C} (x', z')$ to indicate that, under conjugation by a Clifford gate $C$, a single-qudit Pauli operator $X^x Z^z$ is mapped to $X^{x'} Z^{z'}$. Focusing on a single row of the tableau and ignoring the phase, the action is summarized as follows:
\begin{enumerate}
\item[1.] Action of the $\mathcal{F}$ gate on the $i$-th qudit: $(x_{i},z_{i}) \xrightarrow{\mathcal{F}} (-z_{i},x_{i})=(d-z_{i},x_{i})$.

\item[2.] Action of the $S'$ gate on the $i$-th qudit: $(x_{i},z_{i}) \xrightarrow{S'} (x_{i},x_{i}+z_{i})$.

\item[3.] Action of the SUM gate with control $i$ and target $j$: $(x_{i},x_{j},z_{i},z_{j}) \xrightarrow{\mathrm{SUM}(i,j)} (x_{i},x_{i}+x_{j},z_{i}-z_{j},z_{j})$.
\end{enumerate}

Note that Clifford conjugation rules on the stabilizer tableau obey the following elementary properties.

\begin{property}\label{property H}
For any $x,z\in \mathbb{F}_d$, we have $(x,0)\xrightarrow{\mathcal{F}}(0,x)$ and $(0,z)\xrightarrow{\mathcal{F}}(d-z,0)$.
\end{property}

\begin{property}\label{property S}
Assume that $d>2$ is a prime, and fix a non‑zero integer $l \in \{1, 2, \dots, d-1\}$. The set $\{(l, z) : z \in \mathbb{F}_d\}$ forms a single $d$-cycle under repeated application of $S'$. Equivalently, for any distinct $z_1, z_2 \in \mathbb{F}_d$, there exists an integer $k \in \{1, 2, \dots, d-1\}$ such that
\begin{equation}
(l, z_1) \xrightarrow{S'^k} (l, z_2).
\end{equation}
\end{property}

Combining Property~\ref{property H} and Property~\ref{property S}, we can map any pair $(x, z)$ to the canonical form $(1, 0)$.

\begin{property}\label{property HS}
Assume that $d>2$ is a prime. Then for any non‑zero pair $(x, z) \in \mathbb{F}_d^2$ (i.e., $x$ and $z$ not both zero), there exists integers $k_1, k_2 \in \{1, 2, \dots, d-1\}$ such that conjugation by the Clifford gate $\mathcal{F}^2 S'^{k_2} \mathcal{F} S'^{k_1}$ maps $(x, z)$ to the canonical pair $(1, 0)$.
\begin{equation}
(x,z)\xrightarrow{S'^{k_1}}(x,1)\xrightarrow{\mathcal{F}}(d-1,x)\xrightarrow{S'^{k_2}}(d-1,0)\xrightarrow{\mathcal{F}}(0,d-1)\xrightarrow{\mathcal{F}}(1,0).
\end{equation}
\end{property}

The properties established above imply the existence of a Clifford circuit whose action on the phase space $\mathbb{F}_d^{2n}$ sends any $r$-dimensional isotropic subspace with $r \leq n$ to the canonical subspace $0^{2n-r} \times \mathbb{F}_d^r$. An explicit construction of such a circuit for qubit systems was recently presented in~\cite{GIKL25}. In Theorem~\ref{theorem circuit}, we develop the analogous construction rules tailored to qudit systems. The detailed proof is given below.

\begin{proof}[Proof of Theorem~\ref{theorem circuit}]
First, applying Gaussian elimination to the given collection of $m$ vectors yields a basis for the subspace they span. The resulting basis can be arranged as an $r \times 2n$ matrix $M:=(x_{ij},z_{ij})_{i\leq r;j\leq 2n}$, which is in row echelon form. By construction, the space spanned by its row vectors is exactly $A\subseteq \mathbb{F}_d^{2n}$.

Since $M$ can be interpreted as a stabilizer tableau, we can apply Clifford gates to implement elementary row operations on $M$ via conjugation. Our objective is to transform $M$ into the canonical form $(0~ |~ 0~ I_r)$, whose row space is precisely the subspace $0^{2n-r} \times \mathbb{F}_d^r$.

{\bf Step 1.} For each row $i\in\{1,2,\dots,r\}$ of $M$, we perform the following operations:
\begin{enumerate}
\item[(a)] For every $j\in\{i,i+1,\dots,n\}$, apply a sequence of $\mathcal{F}$ and $S'$ gates, as described in Property~\ref{property HS}, to ensure that either $x_{ij}=1, z_{ij}=0$ or $x_{ij}=z_{ij}=0$.

\item[(b)] If $x_{ii}=0$, there exists a $k\in \{i+1,i+2,\dots,n\}$ such that $x_{ik}=1$. Apply the $\mathrm{SUM}(k,i)$ gate such that $x_{ii}=1$.

\item[(c)] For each $j\in \{i+1,i+2,\dots,n\}$, if $x_{ij}=1$, apply the $\mathrm{SUM}^{d-1}(i,j)$ gate, which yields $x_{ij}=1+(d-1)x_{ii}=0$.

\item[(d)] For each $j\in \{i+1,i+2,\dots,r\}$, $x_{ji}=0$ by adding the $i$-th row to the $j$-th row.
\end{enumerate}

{\bf Step 2.} Apply $\mathcal{F}^{\otimes r}$ to the first $r$ qudits.

{\bf Step 3.} Apply the following sequence of gates to the entire quantum system: $ \prod_{i=0}^{r-1}\mathrm{SUM}(n-i,r-i)^{d-1} $ followed by $\prod_{i=0}^{r-1}\mathrm{SUM}(r-i,n-i)$.

The effect of these three steps can be summarized as
\begin{equation}
M \xrightarrow{\text{Step 1}} (I_r~ 0~ |~ 0) \xrightarrow{\text{Step 2}} (0~ |~ I_r~ 0) \xrightarrow{\text{Step 3}} (0~ |~ 0~ I_r),
\end{equation}
where Step 3 can be elaborated as follows: for $r < n/2$,
\begin{equation}
(0~ |~ I_r~ 0) \xrightarrow{ \prod_{i=0}^{r-1}\mathrm{SUM}(n-i,r-i)^{d-1} } (0~ |~ I_r~ 0~ I_r)\xrightarrow{\prod_{i=0}^{r-1}\mathrm{SUM}(r-i,n-i)} (0~ |~ 0~ I_r),
\end{equation}
and for $n/2 \leq r \leq n$,
\begin{equation}
(0~ |~ I_r~ 0) \xrightarrow{ \prod_{i=0}^{r-1}\mathrm{SUM}(n-i,r-i)^{d-1} } (0~|~A_{r\times n})
\xrightarrow{\prod_{i=0}^{r-1}\mathrm{SUM}(r-i,n-i)} (0~ |~ 0~ I_r),
\end{equation}
where $A_{r \times n}=\scriptsize{\begin{pmatrix}
  1&  \cdots  &1 &  &  &   \\
  &  1& \cdots  & 1 &  &   \\
  &  &  \ddots &  & \ddots  &   \\
  &  &  &  1& \cdots  &1  
\end{pmatrix}}.$

The gate complexity of the overall procedure is dominated by the operations performed in Step 1. For each of the $r$ rows, we apply $O(n)$ Clifford gates, which yields a total gate count of $O(rn)$. Since $r \leq \min(m, n)$, the overall running time is ultimately bounded by the cost of Gaussian elimination, namely $O(mn \cdot \min(m, n))$.
\end{proof}

\section{Estimating correlations via the SWAP test}\label{SWAP test}

The SWAP test is a standard quantum algorithm for estimating the overlap between two quantum states by measuring an auxiliary control qubit after a controlled swap operation. When the ancilla remains a qubit, the SWAP test procedure between two qudit states proceeds exactly as the standard qubit protocol~\cite{BCW+01}.

\begin{lemma}\label{lemma SWAPTest}
For any two $n$-qudit pure states $\ket{\psi}$ and $\ket{\phi}$, the probability of obtaining the outcome $0$ in the SWAP test depicted in Figure~\ref{Fig SWAP Test} is given by $p(0)=\frac12(1+\left|\braket{\psi|\phi}\right|^2)$.
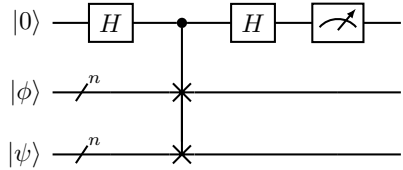
\begin{figure}[ht]
	\centering
	\begin{adjustbox}{width=0.3\textwidth}
	\begin{quantikz}
		\lstick{$ \ket{0} $} &  \gate{H} & \ctrl{2} & \gate{H} & \meter{} &\qw  \\
		\lstick{$ \ket{\phi} $} & \qwbundle{n} & \swap{1} & \qw &\qw &\qw \\
		\lstick{$ \ket{\psi} $} & \qwbundle{n} & \targX{} & \qw &\qw &\qw
	\end{quantikz}
	\end{adjustbox}
	\caption{Quantum circuit for SWAP test of two $n$-qudit states.}
	\label{Fig SWAP Test}
\end{figure}
\end{lemma}

Based on the explicit expression for $p(0)$ in Lemma~\ref{lemma SWAPTest}, we can invert the linear relationship between the measurement outcome probability and the fidelity to estimate $\left|\braket{\psi|\phi}\right|^2$ from repeated SWAP tests. The following theorem formalizes this procedure and analyzes its sample and time complexity.

\begin{theorem}\label{theorem estimation}
For two $n$-qudit pure states $\ket{\psi}$ and $\ket{\phi}$, and for $\varepsilon,\delta>0$, there exists a quantum algorithm that outputs an estimate $E$ of $\left|\braket{\psi|\phi}\right|^2$ such that
\begin{equation}
\Pr\left[\left|E-\left|\braket{\psi|\phi}\right|^2\right| \geq \varepsilon\right] \leq \delta.
\end{equation}
The algorithm requires $O\left(\frac{1}{\varepsilon^2}\log \frac{1}{\delta}\right)$ copies of $\ket{\psi}$ and $\ket{\phi}$, and runs in time $O\left(\frac{n}{\varepsilon^2} \log \frac{1}{\delta}\right)$.
\end{theorem}

\begin{proof}
From Lemma~\ref{lemma SWAPTest}, the probability of obtaining the outcome $0$ in a single SWAP test is given by $p(0)=\frac{1}{2}(1+\left|\braket{\psi|\phi}\right|^2)$. Then, $\left|\braket{\psi|\phi}\right|^2 = 2p(0)-1$.
We repeat the SWAP test $N$ times and introduce independent random variables $X_1, \dots, X_N$, where $X_i = 1$ if the $i$-th outcome is $0$ and $X_i = 0$ otherwise. The empirical average is denoted by $\overline{X} = \frac{1}{N} \sum_{i=1}^N X_i$, and by linearity of expectation we have $\mathbb{E}[\overline{X}] = p(0)$.

To estimate $\left|\braket{\psi|\phi}\right|^2$, we define $E = 2\overline{X}-1$. By Hoeffding's inequality, for any $\gamma > 0$,
\begin{equation}
\Pr\left[\,\left|\overline{X} - p(0)\right| \geq \gamma \,\right] \le 2e^{-2N\gamma^2}.
\end{equation}
Choose $\gamma = \varepsilon/2$. If $\left|\overline{X} - p(0)\right| \le \varepsilon/2$, then $\left|E - \left|\braket{\psi|\phi}\right|^2 \right| = 2\left|\overline{X} - p(0)\right| \leq \varepsilon$.
Hence, setting $N = \bigl\lceil \frac{2}{\varepsilon^2}\ln\frac{2}{\delta} \bigr\rceil$ we can guarantee $\Pr\left[\,\left|\overline{X}-p(0)\right| \geq \varepsilon\,\right]\leq \delta$.

Each SWAP test uses two copies of $\ket{\psi}$, so the total sample complexity is $O(N) = O\left(\frac{1}{\varepsilon^2}\log\frac{1}{\delta}\right)$.
Implementing a controlled-SWAP gate takes $O(1)$ time, and one SWAP test requires $n$ such gates, so the time per test is $O(n)$. The total time is $O(nN) = O\left(\frac{n}{\varepsilon^2} \log\frac{1}{\delta}\right)$. 
\end{proof}

Let $W_{\mathbf{x}}$ be a Weyl operator on an $n$-qudit system, and set $\ket{\phi} = W_{\mathbf{x}} \ket{\psi}$. By Theorem~\ref{theorem estimation}, the SWAP test provides an efficient means of estimating the squared overlap $\left|\bra{\psi} W_{\mathbf{x}} \ket{\psi}\right|^2$. Now suppose we wish to estimate this quantity for a collection of $M$ Weyl operators $W_{\mathbf{x}_1}, \dots, W_{\mathbf{x}_M}$. We can then set the allowed failure probability for each individual estimate to $\delta / M$. Applying the union bound, we conclude that with probability at least $1 - \delta$, all $M$ estimates are simultaneously valid. The formal statement is given below.

\begin{corollary}[Theorem~\ref{theorem M estimations}]
Let $d>2$ be a prime. For any $n$-qudit pure state $\ket{\psi}$ and any collection of $M$ Weyl operators $W_{\mathbf{x}_1}, \dots, W_{\mathbf{x}_M}$, there exists an algorithm that, with probability at least $1 - \delta$, estimates $\bigl| \bra{\psi} W_{\mathbf{x}_i} \ket{\psi} \bigr|^2$ for every $i \in \{1, \dots, M\}$ to an additive error at most $\varepsilon$.  Its sample complexity is $O\bigl( \frac{ M}{\varepsilon^2} \log \frac{M}{\delta} \bigr)$, and the corresponding time complexity is $O\bigl( \frac{ M}{\varepsilon^2} n \log \frac{M}{\delta} \bigr)$.
\end{corollary}

\bibliographystyle{IEEEtran}
\bibliography{Learn_qudit}



\ifCLASSOPTIONcaptionsoff
  \newpage
\fi

\begin{IEEEbiographynophoto}{Wentao Qi}
received the M.S. degree in mathematics from Zhejiang University, China, in 2022. He is currently pursuing the Ph.D. degree with the Institute of Quantum Computing and Software, School of Computer Science and Engineering, Sun Yat-sen University. His research interests include quantum computing and quantum algorithms.
\end{IEEEbiographynophoto}
\begin{IEEEbiographynophoto}{Boyan Xu}
received the B.S. degree in Computer Science and Engineering from Sun Yat-sen University, China, in 2022. He is currently pursuing the Ph.D. degree at the Institute of Quantum Computing and Software, School of Computer Science and Engineering, Sun Yat-sen University. His research interests include testing and learning algorithms in quantum computing and quantum circuit complexity.
\end{IEEEbiographynophoto}
\begin{IEEEbiographynophoto}{Shiguang Feng} 
received the B.S. degree in computer science and technology
from Shandong Agricultural University, Tai'an, China, in 2006; the Ph.D. degree in logic from Sun Yat-sen University, Guangzhou, China, in 2012; and the Doctor of Natural Science degree in computer science from Leipzig University, Leipzig, Germany, in 2016. He is currently an associate researcher with the School of Computer Scienceand Engineering, Sun Yat-sen University, Guangzhou, China. His current research interests include reversible logic synthesis, quantum algorithms, and mathematical logic.
\end{IEEEbiographynophoto}
\begin{IEEEbiographynophoto}{Lvzhou Li}
received his Ph.D. degree in Computer Science from Sun Yat-sen University, China in 2009 and then worked at Sun Yat-sen University, China. He is currently a Professor with the School of Computer Science and Engineering, Sun Yat-sen University, China. His research interests are quantum algorithm, quantum circuit synthesis and optimization, and quantum machine learning.
\end{IEEEbiographynophoto}






\end{document}